\documentclass[12pt, draftclsnofoot, onecolumn]{IEEEtran}
\usepackage[IEEEtran]{./research17}

\usepackage[top=0.81in,bottom=0.81in,left=0.885in,right=0.885in]{geometry}

\usepackage{hyperref}
\hypersetup{hidelinks} 
\usepackage{amsmath,amsthm,amssymb,amsfonts,mathrsfs}

\usepackage{setspace}
\usepackage{threeparttable}
\usepackage{diagbox}
\usepackage{array}
\usepackage{boldline}
\usepackage{caption}
\usepackage{bbm}
\usepackage{algorithmicx}
\usepackage{algpseudocode}

\usepackage{float}
\usepackage{caption}

\usepackage{graphicx}
\usepackage{subfigure}
\usepackage{float}

\usepackage{enumerate}
\usepackage{enumitem}

\usepackage{setspace}

\usepackage{bm}
\usepackage{tabularx}
\usepackage{diagbox}
\usepackage{graphicx}
\usepackage{multirow}

\usepackage{footnote}
\usepackage{threeparttable}
\usepackage{subfigure}

\usepackage{cite}

\theoremstyle{definition}
\newtheorem{theorem}{Theorem}

\newtheorem{lemma}{Lemma}

\makeatletter
\def\thanks#1{\protected@xdef\@thanks{\@thanks
		\protect\footnotetext{#1}}}
\makeatother

\newcounter{MYtempeqncnt}

\ifCLASSINFOpdf
\else
\fi

\begin{document}
\title{On the Capacity Region of Optical Intensity Broadcast Channels}

\author{Sufang~Yang,~\IEEEmembership{Graduate Student Member,~IEEE,}
	Longguang~Li,~\IEEEmembership{Member,~IEEE,}
	and~Jintao~Wang,~\IEEEmembership{Senior~Member,~IEEE}
	\thanks{This work was supported in part by the National Natural Science Foundation of China under Grant No. 62101192, in part by Shanghai Sailing Program under Grant No. 21YF1411000, and in part by Tsinghua University-China Mobile Research Institute Joint Innovation Center.}
	\thanks{Sufang~Yang is with the Department of Electronic Engineering, Tsinghua University, Beijing 100084, China (e-mail: ysf20@mails.tsinghua.edu.cn).}
	\thanks{Longguang Li is with the Shanghai Key Laboratory of Multidimensional Information Processing, and the Department of Communication and Electronic Engineering, East China Normal University, Shanghai 200241, China (e-mail: lgli@cee.ecnu.edu.cn).}
	\thanks{Jintao~Wang is with the Beijing National Research Center for Information Science and Technology (BNRist), Tsinghua University, Beijing 100084, China, also with the Department of Electronic Engineering, Tsinghua University, Beijing 100084, China, and also with the Research Institute, Tsinghua University in Shenzhen, Shenzhen 518057, China (e-mail: wangjintao@tsinghua.edu.cn).}
}

\maketitle

\vspace{-1.7cm}
\begin{abstract}
This paper investigates the capacity region of the optical intensity broadcast channels (OI-BCs), where the input is subject to a peak-intensity constraint, an average-intensity constraint, or both. By leveraging the decomposition results of several random variables, i.e., uniform, exponential, and truncated exponential random variables, and adopting a superposition coding (SC) scheme, the inner bound on the capacity region is derived. Then, the outer bound is derived by applying the conditional entropy power inequality (EPI). In the high signal-to-noise ratio (SNR) regime, the inner bound asymptotically matches the outer bound, thus characterizing the high-SNR asymptotic capacity region. The bounds are also extended to the general $K$-user BCs without loss of high-SNR asymptotic optimality.
\end{abstract}


\section{Introduction}
\label{sec:introduction}
\IEEEPARstart{R}{ecent} years have witnessed rapid development and improvement of wireless communication technology, which has brought the vast popularity of the application of smartphones, intelligent robots, and driverless cars \cite{Yang2019,Giordani,Chowdhury2020}. These mobile terminals provide great convenience to our daily lives. With the increasing number of mobile terminals and the demand for stable and high-speed data transmission, the traditional radio frequency (RF) wireless spectrum is facing severe congestion. To cope with this problem, we urgently need an emerging wireless communication to compensate for the shortcomings of traditional RF communication. 

Optical wireless communication (OWC), with the advantages of low cost and abundant bandwidth, has attracted extensive attention \cite{Khalighi2014,Kaushal2017,Pathak2015,Karunatilaka2015}. In an OWC system, information is carried on the intensity of optical light emitted by light-emitting diodes (LEDs) or laser diodes (LDs) and detected by photodetectors, which is known as the intensity-modulation and direct-detection (IM-DD) scheme \cite{Cox1997,Lee2009}. This unique scheme leads to fundamental differences between the OWC and RF systems. Here, the input corresponds to the optical intensity signal. Hence it is real-valued and non-negative. Furthermore, the peak or/and average intensity are often constrained for the input with battery limitation and safety considerations \cite{Lapidoth2009,Chaaban2016-1,Zhou2017,Chaaban2017-1}. With the advancement of coding, modulation, and detection technologies, the OWC systems have been successfully applied in offices, hospitals, and airport cabins \cite{Fath2013,Uysal2014,Ghassemlooy2015}. 

Generally, the OWC systems can be classified as outdoor free-space optical communication (FSO) and indoor visible light communication (VLC) systems \cite{Chan2006,Komine2004,Elgala2011}. In this paper, we mainly focus on the indoor VLC systems and explore them from an information-theoretic perspective. Extensive studies have been done on the channel capacity for single-user indoor VLC systems \cite{Lapidoth2009,Chaaban2016-1,McKellips2004,Hranilovic2004,Farid2010,Chaaban2018,Li2020}. For example, the authors in \cite{Lapidoth2009} derived the upper and lower bounds on the single-user channel capacity when the input satisfies peak or/and average-intensity constraints. Besides, the asymptotic capacity results at low and high signal-to-noise ratios (SNRs) were also proposed. In recent years, there has been an increasing interest in the multi-user VLC networks, where the transmitter serves multiple users so that each user reliably receives messages simultaneously \cite{Gamal2011}.

Considering the input satisfies peak- or/and average-intensity constraints, several studies have been done on the capacity region of the VLC networks, such as broadcast channels (BCs), multiple access channels (MACs), and interference channels (ICs) \cite{Chaaban2016-2,Chaaban2017-2,Zhou2019,Zhang2022}. For the BC network, the authors in \cite{Chaaban2016-2} considered both peak- and average-intensity constraints for the input. They derived the inner bound on the capacity region based on the truncated Gaussian or on-off keying (OOK) coding scheme and the outer bound based on Bergmans' approach. These bounds asymptotically matched at low SNR and achieved a constant gap at high SNR. For the MAC network, the authors in \cite{Chaaban2017-2} adopted a similar method as in \cite{Chaaban2016-2} and the derived bounds achieved similar performances. Different from \cite{Chaaban2017-2}, the authors in \cite{Zhou2019} considered the peak- or average-intensity constraint for the input. They derived the inner bound based on the uniform or exponential coding scheme, which achieved asymptotic optimality at high SNR. For the IC network, the authors in \cite{Zhang2022} considered both peak- and average-intensity constraints for the input. They derived the inner bound based on the treating-interference-as-noise (TIN) or Han-Kobayashi (HK) schemes and the outer bounds by providing different side information to the receivers.

In this paper, we investigate the optical intensity BCs (OI-BCs). Different from \cite{Chaaban2016-2}, here we consider the input is subject to three different intensity constraints: (1) only peak-intensity constraint; (2) only average-intensity constraint; (3) both peak- and average-intensity constraints. We provide new inner and outer bounds on the capacity region of the considered OI-BCs. Instead of the truncated Gaussian codes \cite{Chaaban2016-2}, we adopt a superposition coding (SC) scheme from the decomposition results of several random variables to derive the inner bound. These random variables are the maximal entropy-achieving inputs when considering peak- or/and average-intensity constraints \cite{Lapidoth2009}. The outer bound is derived by applying the conditional entropy power inequality (EPI) \cite{Verdu2006}. We show that the bounds are asymptotically optimal in the high-SNR regime. Moreover, we extend the bounds to the $K$-user OI-BCs without loss of high-SNR asymptotic optimality.

The remaining part of this paper is organized as follows. Sec.~\ref{sec:channel-model} introduces the channel model of the OI-BCs. Secs.~\ref{sec: two-user capacity region} and~\ref{sec: K-user capacity region} analyze the capacity regions of two-user and $K$-user OI-BCs, respectively. We conclude this paper in Sec.~\ref{sec:conclusion}. A few proofs are given in the Appendices.

\textbf{Notation:} A random variable is denoted by an uppercase letter, e.g., $X$, while its realization by a smallcase letter, e.g., $x$. The support of a random variable is denoted by $\mathsf{supp}(\cdot)$, the expectation by $\mathbb{E}[\cdot]$, and the characteristic function by $\phi(\cdot)$. Differential entropy is denoted by $\mathsf{h}(\cdot)$ and mutual information by $\mathsf{I}(\cdot;\cdot)$. The set of non-negative real numbers is denoted by $\mathbb{R}^+$ and the set of positive natural numbers by $\mathbb{N}^{+}$. Dirac delta function is denoted by $\delta (\cdot)$ and the convex hull of a set by $\mathsf{Conv}\{ \cdot\}$. The index set $\{1,2,\ldots,k\}$ is denoted by $[k]$, $k\in\mathbb{N}^{+}$, and $j=\sqrt{-1}$. Function $A(t)\doteq B(t)$ denotes that $\lim_{t\rightarrow\infty} \{A(t)-B(t)\}=0$.


\section{Channel Model}\label{sec:channel-model}
For a single-user optical intensity channel, the channel output is given by
{\setlength\abovedisplayskip{4.85pt} 
\setlength\belowdisplayskip{4.85pt}
\begin{IEEEeqnarray}{rCl}
	Y =X+Z,	\label{eq:singlechannelmodel}
\end{IEEEeqnarray}}where $X$ denotes channel input and  $Z$ denotes Gaussian noise with variance $\sigma^2$, i.e., 
{\setlength\abovedisplayskip{4.85pt} 
\setlength\belowdisplayskip{4.85pt}
\begin{IEEEeqnarray}{c}
	Z\sim \mathcal{N}(0,\sigma^2).
	\label{gaussiannoise}
\end{IEEEeqnarray}}Since $X$ is proportional to optical intensity, its support must satisfy
{\setlength\abovedisplayskip{4.85pt} 
	\setlength\belowdisplayskip{4.85pt}
\begin{IEEEeqnarray}{c}
	\mathsf{supp} (X) \subset \mathbb{R}^+.
	\label{eq: positive support}
\end{IEEEeqnarray}}Considering the limited dynamic range of LED devices and the requirement of illumination quality or energy consumption, the input is usually subject to a peak-intensity constraint:
{\setlength\abovedisplayskip{4.85pt} 
\setlength\belowdisplayskip{4.85pt}
\begin{IEEEeqnarray}{c}
	\textnormal{Pr}( X \leq \amp ) = 1, \label{eq:peak cons}
\end{IEEEeqnarray}}or/and an average-intensity constraint:
{\setlength\abovedisplayskip{4.85pt} 
\setlength\belowdisplayskip{4.85pt}
\begin{IEEEeqnarray}{c}
	\mathbb{E} [X]  \leq \EE. 
	\label{eq:ave cons}
\end{IEEEeqnarray}}For convenience, we denote the ratio between the maximal instantaneous intensity and average intensity as $\alpha$, i.e., $\alpha=\frac{\EE}{\amp}$, which is limited in $\left(0,\frac{1}{2}\right]$.

In this paper, we mainly focus on the OI-BCs. We first consider a two-user OI-BC, and the channel output at user $k$ is given by
{\setlength\abovedisplayskip{4.85pt} 
\setlength\belowdisplayskip{4.85pt}
\begin{IEEEeqnarray}{rCl}
	Y_k =X+Z_k,\quad k\in\{1,2\},
	\label{eq:channelmodel}
\end{IEEEeqnarray}}where $X$ denotes the channel input, whose support satisfies~\eqref{eq: positive support}, and is also subject to~\eqref{eq:peak cons}, or/and~\eqref{eq:ave cons}; and $Z_k$ denotes the Gaussian noise at user $k$, and 
{\setlength\abovedisplayskip{4.85pt} 
\setlength\belowdisplayskip{4.85pt}
\begin{IEEEeqnarray}{rCl}
	Z_k\sim \mathcal{N}(0,\sigma_k^2), \quad k\in\{1,2\}.
\end{IEEEeqnarray}}Without loss of generality, we assume $\sigma_1 \leq \sigma_2$. Besides, a high SNR regime corresponds to the regime where $\amp\gg\sigma_k$ or $\EE\gg\sigma_k$, $k\in\{1,2\}$.

In a two-user OI-BC, the desired messages for users 1 and 2 are denoted by $M_1$ and $M_2$, respectively. We assume $M_1$ and $M_2$ are independent. The transmitter encodes and sends them at coding rates $R_1$, and $R_2$, respectively, and simultaneously. If both users can decode their messages with vanishing error probabilities as the code length tends to infinity, we say that the rate pair $(R_1,R_2)$ is achievable. The capacity region of this channel is the closure of the set of all achievable rate pairs.

Note that the channel in~\eqref{eq:channelmodel} belongs to the class of the additive white Gaussian broadcast channels, which is a physically degraded channel \cite{Gamal2011}. As a consequence, we have
{\setlength\abovedisplayskip{4.85pt} 
\setlength\belowdisplayskip{4.85pt}
\begin{IEEEeqnarray}{rCl}
	Y_2 = Y_1 + \widetilde{Z}_2,
\end{IEEEeqnarray}}where $\widetilde{Z}_2\sim\mathcal{N}(0,\sigma_2^2-\sigma_1^2)$. The capacity region of a two-user physically degraded channel is characterized by~\cite[Theorem 15.6.2]{Cover2006} and summarized in the following lemma:
\begin{lemma}\label{lemma: twouser}
	The capacity region of the two-user OI-BC in~\eqref{eq:channelmodel}, where $X\rightarrow Y_1 \rightarrow Y_2$ forms a Markov chain, is the set of rate pair $(R_1,R_2)$ satisfying
	{\setlength\abovedisplayskip{4.85pt} 
	\setlength\belowdisplayskip{4.85pt}
	\begin{IEEEeqnarray}{rCl}
		R_1 &\leq& \mathsf{I}(X;Y_1|U), \\
		R_2 &\leq& \mathsf{I}(U;Y_2),
	\end{IEEEeqnarray}}for some $p_{U,X}(u,x)$ subject to~\eqref{eq: positive support},~\eqref{eq:peak cons}, or/and~\eqref{eq:ave cons}.
\end{lemma}

Above notations and assumptions can be directly extended to the $K$-user OI-BC, i.e.,
{\setlength\abovedisplayskip{4.85pt} 
\setlength\belowdisplayskip{4.85pt}
\begin{IEEEeqnarray}{rCl}
	Y_k = X + Z_k, \quad k\in[K], \label{K-user channel model}
\end{IEEEeqnarray}}where $X$ is a codeword of independent messages $M_1,M_2\cdots,M_K$, and satisfies \eqref{eq: positive support},~\eqref{eq:peak cons}, or/and~\eqref{eq:ave cons}; and $Z_k\sim\mathcal{N}(0,\sigma_k^2)$, $k\in[K]$. Without loss of generality, we also assume $\sigma_1 \leq \sigma_2 \leq \ldots \leq \sigma_K$. It should be noted that the above $K$-user OI-BC is also a physically degraded channel, and we have
{\setlength\abovedisplayskip{4.85pt} 
\setlength\belowdisplayskip{4.85pt}
\begin{IEEEeqnarray}{rCl}
	Y_{k+1} = Y_k + \widetilde{Z}_{k+1}, \quad k \in [K-1],
\end{IEEEeqnarray}}where $\widetilde{Z}_{k+1}\sim\mathcal{N}(0,\sigma_{k+1}^2-\sigma_{k}^2)$. The capacity region of a $K$-user physically degraded channel is characterized by~\cite[Chapter 5.7]{Gamal2011}, \cite[Theorem 3]{Nair2011} and summarized in the following lemma:
\begin{lemma}\label{lemma: Kuser}
	The capacity region of the $K$-user OI-BC in \eqref{K-user channel model}, where $X\rightarrow Y_1 \rightarrow Y_2 \rightarrow \ldots \rightarrow Y_K$ forms a Markov chain, is the set of rate tuple $(R_1,R_2,\ldots,R_K)$ satisfying
	{\setlength\abovedisplayskip{4.85pt} 
	\setlength\belowdisplayskip{4.85pt}
	\begin{IEEEeqnarray}{rCl}
		R_k &\leq& \mathsf{I}(U_k;Y_k|U_{k+1}), \quad k \in [K],
	\end{IEEEeqnarray}}for some $p_{U_{K-1}|U_K}(u_{K-1}|u_{K})p_{U_{K-2}|U_{K-1}}(u_{K-2}|u_{K-1})\cdots p_{U_1|U_2}(u_1|u_2)$ subject to~\eqref{eq: positive support},~\eqref{eq:peak cons}, or/and~\eqref{eq:ave cons}, where $U_1=X$, $U_{K+1}=\emptyset$, .
\end{lemma}

\section{Capacity Region Characterization of Two-User OI-BCs}\label{sec: two-user capacity region}
In this section, we characterize the capacity region of two-user OI-BCs with three different input constraints. In each case of input constraints, we first present some preliminaries about the decomposition of the random variable and the existing single-user channel capacity, then derive the inner and outer bounds on the capacity region. Finally, the high-SNR capacity region is characterized based on these bounds.

\subsection{Peak-Intensity Constrained OI-BC}
\subsubsection{Preliminary}
Considering the peak-intensity constraint in \eqref{eq:peak cons}, the decomposition of a uniform random variable and the existing single-user channel capacity is introduced in the following.

\begin{lemma}[\cite{Li2022}]\label{lemma: decom of uniform}
	Given any integer $N\in\mathbb{N}^{+}$, a uniform random variable $S$, i.e.,
	{\setlength\abovedisplayskip{4.85pt} 
	\setlength\belowdisplayskip{4.85pt}
	\begin{IEEEeqnarray}{rCl}
		p_{S}(s) = \frac{1}{\amp},\quad s\in[0,\amp], \label{eq: uniform distribution}
	\end{IEEEeqnarray}}can be decomposed as a sum of two independent random variables $S_1$ and $S_2$, i.e.,
	{\setlength\abovedisplayskip{4.85pt} 
	\setlength\belowdisplayskip{4.85pt}
	\begin{IEEEeqnarray}{rCl}
		p_{S_1}(s_1) &=& \frac{N}{\amp},\quad s_1 \in \left[ 0,\frac{\amp}{N} \right],\\
		p_{S_2}(s_2) &=& \frac{1}{N} \sum_{n=0}^{N-1} \delta \left(s_2-\frac{n}{N}\times\amp \right).
	\end{IEEEeqnarray}}	
\end{lemma}For convenience, the distribution in \eqref{eq: uniform distribution} is denoted by $\textsf{U}[0,\amp]$.

\begin{lemma}[\cite{McKellips2004,Lapidoth2009}]\label{lemma: 1-user capacity only peak}
	The channel capacity $\mathsf{C}(\amp,\sigma)$ of a peak-intensity constrained optical intensity channel, is upper bounded by
	{\setlength\abovedisplayskip{4.85pt} 
	\setlength\belowdisplayskip{4.85pt}
	\begin{IEEEeqnarray}{rCl}
	\mathsf{C}_{\textnormal{ub}}(\amp,\sigma) = \min \left\{\frac{1}{2}\log \left(1+\frac{\amp^2}{4\sigma^2} \right),\,\log \left(1+ \frac{\amp}{\sqrt{2\pi e}\sigma}\right) \right\}, \label{eq: 1-user ub capacity only peak}
	\end{IEEEeqnarray}}At high SNR,\footnote{It should be noted that $\mathsf{C}_{\textnormal{ub}}(\amp,\sigma)$ can be further simplified to $\mathsf{C}_{\textnormal{ub}}(\amp,\sigma)\doteq \frac{1}{2}\log \left(\frac{\amp^2}{2\pi e \sigma^2} \right)$, but we keep the current form in \eqref{eq: only peak high-SNR capacity} to simplify the subsequent calculations.}
	{\setlength\abovedisplayskip{4.85pt} 
	\setlength\belowdisplayskip{4.85pt}
	\begin{IEEEeqnarray}{rCl}
		\mathsf{C}_{\textnormal{ub}}(\amp,\sigma) \doteq \frac{1}{2}\log \left(1+\frac{\amp^2}{2\pi e \sigma^2} \right). \label{eq: only peak high-SNR capacity}
	\end{IEEEeqnarray}}
\end{lemma}

\subsubsection{Bounds on Capacity Region}
The inner bound on the capacity region is derived based on the SC scheme.
\begin{theorem}[\textbf{Inner Bound}]\label{theorem: inner bound only peak}
	When the input is only subject to the peak-intensity constraint in \eqref{eq:peak cons}, the rate pairs in the set $\mathsf{Conv}\{ \cup_{N\in\mathbb{N}^{+}} \left( R_{1}^{\textnormal{in}}(N),R_{2}^{\textnormal{in}}(N) \right) \}$ are all achievable for a two-user OI-BC, where
	{\setlength\abovedisplayskip{4.85pt} 
	\setlength\belowdisplayskip{4.85pt}
	\begin{IEEEeqnarray}{rCl}
		\subnumberinglabel{eq: only peak IB}
		R_{1}^{\textnormal{in}}(N) &=& \frac{1}{2} \log \left( 1+\frac{\amp^2}{2 \pi e N^2\sigma_1^2} \right),\label{eq: R1 Lb only peak}\\
		R_{2}^{\textnormal{in}}(N) &=& \frac{1}{2} \log \left( 1+\frac{\amp^2}{2 \pi e \sigma_2^2} \right) - \mathsf{C}_{\textnormal{ub}} \left(\frac{\amp}{N}, \sigma_2\right). \label{eq: R2 Lb only peak}
	\end{IEEEeqnarray}}
\end{theorem} 
\begin{proof}
	We first encode the messages $M_1$ and $M_2$ independently into signals $X_1$ and $X_2$, where $X_1$ follows $\textsf{U}[0,\frac{\amp}{N}]$ and $X_2$ follows
	{\setlength\abovedisplayskip{4.85pt} 
	\setlength\belowdisplayskip{4.85pt}
	\begin{IEEEeqnarray}{rCl}
		p_{X_2}(x_2) &=& \frac{1}{N} \sum_{n=0}^{N-1} \delta \left(x_2-\frac{n}{N}\times\amp \right).
	\end{IEEEeqnarray}}Then we adopt an SC scheme such that $X=X_1+X_2$. By applying Lemma~\ref{lemma: decom of uniform}, we obtain that $X$ follows $\textsf{U}[0,\amp]$, which satisfies the peak-intensity constraint in \eqref{eq:peak cons}. In Lemma \ref{lemma: twouser}, we instantiate $U$ into $U=X_2$. Therefore, we can compute the achievable rates $\mathsf{I}(X;Y_1|X_2)$ and $\mathsf{I}(X_2;Y_2)$ to be as the inner bounds on $R_1$ and $R_2$, respectively. On the one hand,
	{\setlength\abovedisplayskip{4.85pt} 
	\setlength\belowdisplayskip{4.85pt}
	\begin{IEEEeqnarray}{rCl}
		\mathsf{I}(X;Y_1|X_2) 
		&=& \mathsf{I}(X_1;X_1+Z_1) \label{eq: 21}\\
		&=& \mathsf{h}(X_1+Z_1) - \mathsf{h}(Z_1)\\		
		&\geq& \frac{1}{2} \log \left( 1 + \frac{ e^{2\mathsf{h}(X_1)} }{ e^{2\mathsf{h}(Z_1)} }\right)\label{eq: EPI}\\
		&=& \frac{1}{2} \log \left( 1+\frac{\amp^2}{2 \pi e N^2\sigma_1^2} \right). \label{eq: 65}
	\end{IEEEeqnarray}}where \eqref{eq: EPI} holds by the EPI. On the other hand,
	{\setlength\abovedisplayskip{4.85pt} 
	\setlength\belowdisplayskip{4.85pt}
	\begin{IEEEeqnarray}{rCl}
		\mathsf{I}(X_2;Y_2) &=& \mathsf{h}(Y_2) - \mathsf{h}(Y_2|X_2)\\
		&=& \mathsf{h}(X+Z_2) - \mathsf{h}(X_1+Z_2)\\
		&\geq& \frac{1}{2} \log \left( e^{2\mathsf{h}(X)} + e^{2\mathsf{h}(Z_2)} \right) - \mathsf{h}(X_1+Z_2)\label{eq2: EPI}\\
		&\geq& \frac{1}{2} \log \left( 1+\frac{\amp^2}{2 \pi e \sigma_2^2} \right) - \mathsf{C}_{\textnormal{ub}} \left(\frac{\amp}{N}, \sigma_2\right). \label{eq: 69}
	\end{IEEEeqnarray}}where \eqref{eq2: EPI} holds by the EPI; and \eqref{eq: 69} holds since $X_1$ is limited in $[0,\frac{\amp}{N}]$ and independent of the Gaussian noise $Z_2$, thus we can bound $\mathsf{h}(X_1+Z_2)$ by $\mathsf{h}(X_1+Z_2) \leq \mathsf{C}_{\textnormal{ub}} \left( \frac{\amp}{N},\sigma_2 \right) + \mathsf{h}(Z_2)$.

Combining \eqref{eq: 65} and \eqref{eq: 69}, the rate pairs $\left( R_{1}^{\textnormal{in}}(N),R_{2}^{\textnormal{in}}(N) \right)$ in \eqref{eq: only peak IB} are achievable. Moreover, by adopting a time sharing strategy between these achievable rate pairs, the rate pairs in the set $\mathsf{Conv}\{ \cup_{N\in\mathbb{N}^{+}} \left( R_{1}^{\textnormal{in}}(N),R_{2}^{\textnormal{in}}(N) \right) \}$ are all achievable. Finally, the proof is completed.	
\end{proof}

The outer bound on the capacity region is derived based on the conditional EPI.
\begin{theorem}[\textbf{Outer Bound}]\label{theorem: outer bound only peak}
	When the input is only subject to the peak-intensity constraint in \eqref{eq:peak cons}, the capacity region is outer bounded by $\cup_{\rho\in[0,1]} \left( R_{1}^{\textnormal{out}}(\rho),R_{2}^{\textnormal{out}}(\rho) \right)$ for a two-user OI-BC, where
	{\setlength\abovedisplayskip{4.85pt} 
	\setlength\belowdisplayskip{4.85pt}
	\begin{IEEEeqnarray}{rCl}
		\subnumberinglabel{eq: only peak OB}
		R_{1}^{\textnormal{out}}(\rho) &=& \frac{1}{2} \log \left\{1+\frac{\sigma_2^2}{\sigma_1^2}\left(e^{2 \mathsf{C}_{\textnormal{ub}}\left(\rho \amp, \sigma_2\right)}-1\right)\right\},\label{eq: R1 Ub only peak}\\
		R_{2}^{\textnormal{out}}(\rho) &=& \mathsf{C}_{\textnormal{ub}}\left(\amp, \sigma_2\right)-\mathsf{C}_{\textnormal{ub}}\left(\rho \amp, \sigma_2\right).\label{eq: R2 Ub only peak}
	\end{IEEEeqnarray}}
\end{theorem} 
\begin{proof}
We resort to Lemma \ref{lemma: twouser} to derive an upper bound on the achievable rate $R_2$ of user 2 first and then the achievable rate $R_1$ of user 1. When $U \rightarrow X \rightarrow Y_2$ forms a Markov chain, it follows that
{\setlength\abovedisplayskip{4.85pt} 
\setlength\belowdisplayskip{4.85pt}
\begin{IEEEeqnarray}{rCl}
	\mathsf{h}(Y_2|U) \geq \mathsf{h}(Y_2|X) = \mathsf{h}(Z_2). \label{eq: h(Y2|U) LB}
\end{IEEEeqnarray}}Furthermore, we have
{\setlength\abovedisplayskip{4.85pt} 
\setlength\belowdisplayskip{4.85pt}
\begin{IEEEeqnarray}{rCl}
	\mathsf{h}(Y_2|U) &\leq& \mathsf{h}(Y_2) \leq \mathsf{C}_{\textnormal{ub}}(\amp,\sigma_2) + \mathsf{h}(Z_2). \label{eq: h(Y2|U) UB}
\end{IEEEeqnarray}}Note that $\mathsf{C}_{\textnormal{ub}}(\amp,\sigma_2)$ is monotonically increasing with respect to $\amp$ and approaches zero when $\amp$ tends to $0$. Combined with \eqref{eq: h(Y2|U) LB} and \eqref{eq: h(Y2|U) UB}, there exists $\rho\in[0,1]$ such that
{\setlength\abovedisplayskip{4.85pt} 
\setlength\belowdisplayskip{4.85pt}
\begin{IEEEeqnarray}{rCl}
	\mathsf{h}(Y_2|U) = \mathsf{C}_{\textnormal{ub}}(\rho\amp,\sigma_2) + \mathsf{h}(Z_2). \label{eq: h(Y2|U) equality}
\end{IEEEeqnarray}}Applying Lemma \ref{lemma: twouser}, the achievable rate $R_2$ can be upper bounded by
{\setlength\abovedisplayskip{4.85pt} 
\setlength\belowdisplayskip{4.85pt}
\begin{IEEEeqnarray}{rCl}
	R_2 &\leq& \mathsf{I}(U;Y_2) \label{eq: 34}\\
	&=& \mathsf{h}(Y_2) - \mathsf{h}(Y_2|U) \label{eq: 35}\\
	&\leq& \mathsf{C}_{\textnormal{ub}}(\amp,\sigma_2)-\mathsf{C}_{\textnormal{ub}}(\rho\amp,\sigma_2), \label{eq: 36}
\end{IEEEeqnarray}}where \eqref{eq: 36} is obtained by substituting $\mathsf{h}(Y_2)\leq \mathsf{C}_{\textnormal{ub}}(\amp,\sigma_2) + \mathsf{h}(Z_2)$ and \eqref{eq: h(Y2|U) equality} into \eqref{eq: 35}. Moreover, the achievable rate $R_1$ can be upper bounded by
{\setlength\abovedisplayskip{4.85pt} 
\setlength\belowdisplayskip{4.85pt}
\begin{IEEEeqnarray}{rCl}
	R_1 &\leq& \mathsf{I}(X;Y_1|U) \\
	&=& \mathsf{h}(Y_1|U) - \mathsf{h}(Z_1)\\
	&\leq&  \frac{1}{2} \log \left( e^{2\mathsf{h}(Y_2|U)} - e^{2\mathsf{h}(\widetilde{Z}_2|U)} \right) - \mathsf{h}(Z_1) \label{eq: h(Y1|X)}\\
	&=& \frac{1}{2} \log \left\{1+\frac{\sigma_2^2}{\sigma_1^2}\left(e^{2 \mathsf{C}_{\textnormal{ub}}\left(\rho \amp, \sigma_2\right)}-1\right)\right\}, \label{eq: 40}
\end{IEEEeqnarray}}where \eqref{eq: h(Y1|X)} is obtained by noting that $Y_2 = Y_1 +\widetilde{Z}_2$ and we can utilize the conditional EPI to bound $\mathsf{h}(Y_1|U)$ such that $	\mathsf{h}(Y_1|U) \leq \frac{1}{2} \log \left( e^{2\mathsf{h}(Y_2|U)} - e^{2\mathsf{h}(\widetilde{Z}_2|U)} \right)$; and \eqref{eq: 40} is obtained by substituting \eqref{eq: h(Y2|U) equality} into \eqref{eq: h(Y1|X)}. Combining \eqref{eq: 36} and \eqref{eq: 40}, we complete the proof.
\end{proof}

The high-SNR capacity region is derived based on the inner bound in Theorem~\ref{theorem: inner bound only peak} and the outer bound in Theorem~\ref{theorem: outer bound only peak}. We summarize it as follows.
\begin{theorem}[\textbf{Asymptotic Capacity Region}]\label{theorem: high SNR capacity region only peak}
	When the input is only subject to the peak-intensity constraint in \eqref{eq:peak cons}, at high SNR, the capacity region of a two-user OI-BC asymptotically converges to the region where the rate pair $(R_1,R_2)$ satisfies
	{\setlength\abovedisplayskip{4.85pt} 
	\setlength\belowdisplayskip{4.85pt}
	\begin{IEEEeqnarray}{rCl}
		\subnumberinglabel{eq: only peak high-SNR capacity region}
		R_1 &\ \dot{\leq}\ & \frac{1}{2} \log \left(1+\frac{\rho^2 \amp^2}{2 \pi e \sigma_1^2}\right), \label{eq: high SNR R1 only peak}\\
		R_2 &\ \dot{\leq}\ & \frac{1}{2} \log \left(1+\frac{\left(1-\rho^2\right) \amp^2}{\rho^2 \amp^2+2 \pi e \sigma_2^2}\right), \label{eq: high SNR R2 only peak}
	\end{IEEEeqnarray}}with $\rho \in [0,1]$.
\end{theorem}
\begin{proof}
	At high SNR, the single-user capacity result of peak-intensity constrained optical intensity channel is given in Lemma~\ref{lemma: 1-user capacity only peak}. Substituting it into Theorem~\ref{theorem: outer bound only peak}, we obtain that the capacity region at high SNR is outer bounded by $\cup_{\rho\in[0,1]} \left( R_{1}^{\textnormal{out}}(\rho),R_{2}^{\textnormal{out}}(\rho) \right)$, where
	{\setlength\abovedisplayskip{4.85pt} 
	\setlength\belowdisplayskip{4.85pt}
	\begin{IEEEeqnarray}{rCl}
		\subnumberinglabel{eq: only peak new OB}
		R_{1}^{\textnormal{out}}(\rho) & \doteq & \frac{1}{2} \log \left(1+\frac{\rho^2 \amp^2}{2 \pi e \sigma_1^2}\right),\label{eq: high SNR R1 only peak Ub only peak}\\
		R_{2}^{\textnormal{out}}(\rho) 
		&\doteq& \frac{1}{2} \log \left(1+\frac{\left(1-\rho^2\right) \amp^2}{\rho^2 \amp^2+2 \pi e \sigma_2^2}\right).\label{eq: high SNR R2 only peak Ub only peak}
	\end{IEEEeqnarray}}For convenience, we write the boundary of the above outer bound as 
	{\setlength\abovedisplayskip{4.85pt} 
	\setlength\belowdisplayskip{4.85pt}
\begin{IEEEeqnarray}{rCl}
		R_{2}^{\textnormal{out}} 
		&\doteq& \frac{1}{2} \log \left( 1+\frac{\amp^2}{2 \pi e \sigma_2^2} \right) - \frac{1}{2} \log \left( 1+\frac{\sigma_1^2}{\sigma_2^2}\left(e^{2 R_{1}}-1\right) \right), \ R_1\in\left[ 0,\frac{1}{2} \log \left( 1 +  \frac{\amp^2 }{2 \pi e \sigma_1^2} \right) \right]. \quad
	\end{IEEEeqnarray}}Next, we proceed to show that the inner bound in Theorem~\ref{theorem: inner bound only peak} is tight with \eqref{eq: only peak new OB} at high SNR. Substituting \eqref{eq: only peak high-SNR capacity} into \eqref{eq: only peak IB}, we obtain that the rate pairs in $\mathsf{Conv}\{ \cup_{N\in\mathbb{N}^{+}} \left( R_{1}^{\textnormal{in}}(N),R_{2}^{\textnormal{in}}(N) \right) \}$ are all achievable, where
	{\setlength\abovedisplayskip{4.85pt} 
	\setlength\belowdisplayskip{4.85pt}
	\begin{IEEEeqnarray}{rCl}
		R_{1}^{\textnormal{in}}(N) &\doteq& \frac{1}{2} \log \left( 1 +  \frac{\amp^2 }{2 \pi e N^2\sigma_1^2} \right),\\
		R_{2}^{\textnormal{in}}(N) &\doteq& \frac{1}{2} \log \left(\frac{\amp^2+2 \pi e \sigma_2^2}{\amp^2+2 \pi e N^2 \sigma_2^2} \times N^2\right) .
	\end{IEEEeqnarray}}To analyze the inner bound characterized by the above achievable rate pairs, we fix an integer $N$ and derive the convex hull of two adjacent rate pairs $\left(R_{1}^{\textnormal{in}}(N),R_{2}^{\textnormal{in}}(N)\right)$ and $(R_{1}^{\textnormal{in}}(N+1),R_{2}^{\textnormal{in}}(N+1))$, which is a straight line segment: 
	{\setlength\abovedisplayskip{4.85pt} 
	\setlength\belowdisplayskip{4.85pt}
	\begin{IEEEeqnarray}{rCl}
		R_{2}^{\textnormal{in}}
		&\doteq&\frac{ \log \left(\frac{\amp^2+ 2 \pi e N^2 \sigma_2^2}{\amp^2+ 2 \pi e (N+1)^2 \sigma_2^2} \times \frac{(N+1)^2}{N^2}\right)}{ \log \left(\frac{\amp^2+ 2 \pi e (N+1)^2 \sigma_1^2}{\amp^2+ 2 \pi e N^2 \sigma_1^2} \times \frac{N^2}{(N+1)^2}\right)}
		\left( R_{1} - \frac{1}{2} \log \left(\frac{\amp^2}{  2 \pi e N^2 \sigma_1^2} \right) \right) + \frac{1}{2} \log \left(\frac{\amp^2+2 \pi e \sigma_2^2}{\amp^2+N^2 2 \pi e \sigma_2^2} \times N^2\right)\nonumber\\
		&\doteq& - R_{1} + \frac{1}{2} \log \left(\frac{\amp^2}{ 2 \pi e  \sigma_1^2}\right), \ R_{1}\in[R_{1}^{\textnormal{in}}(N),R_{1}^{\textnormal{in}}(N+1)],\ \forall N\in\mathbb{N}^{+}.
	\end{IEEEeqnarray}}We compute the gap between $R_{2}^{\textnormal{out}}-R_{2}^{\textnormal{in}}$ such that
	{\setlength\abovedisplayskip{4.85pt} 
	\setlength\belowdisplayskip{4.85pt}
	\begin{IEEEeqnarray}{rCl}
		R_{2}^{\textnormal{out}}-R_{2}^{\textnormal{in}} &\doteq &
		\log \left( \frac{\sigma_1^2}{\sigma_2^2} \right)- \frac{1}{2} \log \left( 1+\frac{\sigma_1^2}{\sigma_2^2}\left(e^{2 R_{1}}-1\right) \right) +R_1 \label{eq: gap of R2}\\
		&\doteq &
		\log \left( \frac{\sigma_1^2}{\sigma_2^2} \right)- \frac{1}{2} \log \left( 1+\frac{\rho^2\amp^2}{2\pi e\sigma_2^2} \right) + \frac{1}{2} \log \left( 1+\frac{\rho^2\amp^2}{2\pi e\sigma_1^2} \right) \label{eq: subst R1}\\
		&\doteq& 0, \quad R_1\in[R_{1}^{\textnormal{in}}(N),R_{1}^{\textnormal{in}}(N+1)],\ \forall N\in\mathbb{N}^{+}. \label{eq: zero gap}
	\end{IEEEeqnarray}}where \eqref{eq: subst R1} is obtained by substituting \eqref{eq: high SNR R1 only peak Ub only peak} into \eqref{eq: gap of R2}. When we take all the possible integer $N$ in $\mathbb{N}^{+}$, it follows that
	{\setlength\abovedisplayskip{4.85pt} 
	\setlength\belowdisplayskip{4.85pt}
	\begin{IEEEeqnarray}{rCl}
		R_{2}^{\textnormal{out}}-R_{2}^{\textnormal{in}} &\doteq & 0, \quad \forall R_1\in\left[ 0,\frac{1}{2} \log \left( 1 +  \frac{\amp^2 }{2 \pi e \sigma_1^2} \right) \right], 
	\end{IEEEeqnarray}}which indicates that the outer bound in Theorem~\ref{theorem: outer bound only peak} and inner bound in Theorem~\ref{theorem: inner bound only peak}  are tight at high SNR. Hence, we complete the proof.	
\end{proof}

The derived bounds on the capacity region are shown in Fig.~\ref{fig: two-user only peak}, where we assume $\sigma_2 = 2\sigma_1$ and $\mathrm{ASNR}_k=\frac{\amp}{\sigma_k}$, $k\in\{1,2\}$. From the figure, it is straightforward to observe that the inner and outer bounds become tighter as SNR increases, which numerically validates the derivations.
\begin{figure}[!htbp]
	\centering
	\includegraphics[width=3.7in]{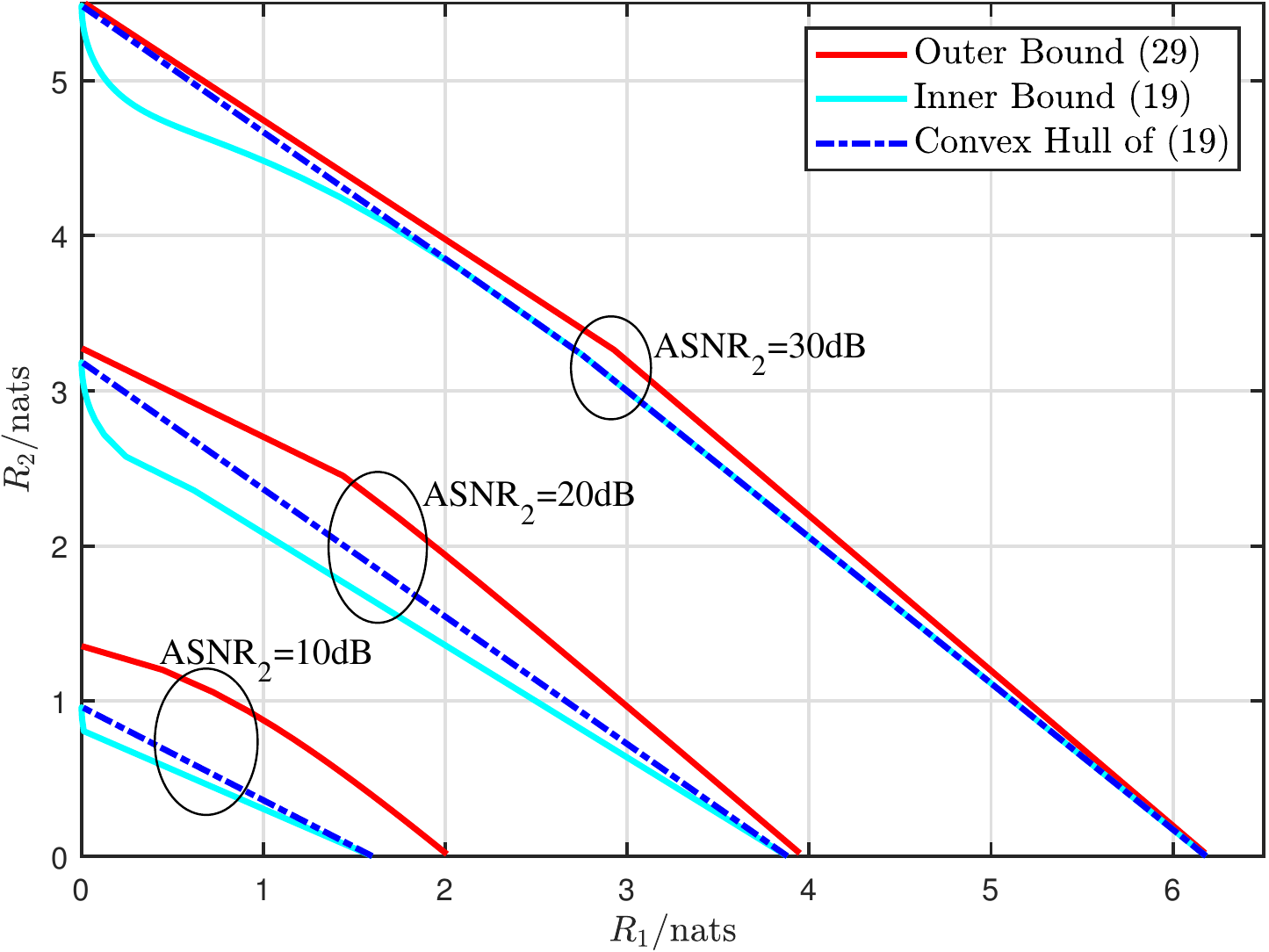}\\
	\caption{Bounds on capacity region of $2$-user OI-BC with peak-intensity constraint.}
	\label{fig: two-user only peak}
	\vspace{-0.7cm}
\end{figure}
\subsection{Average-Intensity Constrained OI-BC}
\subsubsection{Preliminary}
Considering the average-intensity constraint in \eqref{eq:ave cons}, the decomposition of an exponential random variable and the existing single-user channel capacity is introduced in the following.
\begin{lemma}[\cite{Li2022}]\label{lemma: decom of exponential}
	Given any $\EE>0$, $\EE_1>0$, and $\EE_1\leq \EE$, an exponential random variable $S$, i.e.,
	{\setlength\abovedisplayskip{4.85pt} 
	\setlength\belowdisplayskip{4.85pt}
	\begin{IEEEeqnarray}{rCl}
		p_{S}(s) = \frac{1}{\EE} e^{ -\frac{s}{\EE} },\quad s\in[0,+\infty),\label{eq: exp distribution}
	\end{IEEEeqnarray}}can be decomposed as a sum of two independent random variables $S_1$ and $S_2$, i.e.,
	{\setlength\abovedisplayskip{4.85pt} 
	\setlength\belowdisplayskip{4.85pt}
	\begin{IEEEeqnarray}{rCl}
		p_{S_1}(s_1) &=& \frac{1}{\EE_1} e^{ -\frac{s_1}{\EE_1} },\quad s_1\in[0,+\infty),\\
		p_{S_2}(s_2) &=& \frac{\EE_1}{\EE} \delta(s_2) + \left( 1-\frac{\EE_{1}}{\EE} \right) \times \frac{1}{\EE}e^{-\frac{s_2}{\EE}},\quad s_2\in[0,+\infty).
	\end{IEEEeqnarray}}	
\end{lemma}For convenience, the distribution in \eqref{eq: exp distribution} is denoted by $\textsf{Exp}(\EE)$.

\begin{lemma}[\cite{Lapidoth2009,Hranilovic2004,Farid2010}]\label{lemma: 1-user capacity only ave}
	The channel capacity $\mathsf{C}(\EE,\sigma)$ of an average-intensity constrained optical intensity channel, is upper bounded by
	{\setlength\abovedisplayskip{4.85pt} 
	\setlength\belowdisplayskip{4.85pt}
	\begin{IEEEeqnarray}{rCl}
		\mathsf{C}_{\textnormal{ub}}(\EE,\sigma) = \frac{1}{2}\log \left\{ \frac{e}{2\pi} \left( 2+ \frac{\EE}{\sigma} \right)^2 \right\},\label{eq: 1-user ub capacity only ave}
	\end{IEEEeqnarray}}At high SNR,
	{\setlength\abovedisplayskip{4.85pt} 
	\setlength\belowdisplayskip{4.85pt}
	\begin{IEEEeqnarray}{rCl}
		\mathsf{C}_{\textnormal{ub}}(\EE,\sigma) \doteq \frac{1}{2}\log \left(1+ \frac{ e\EE^2 }{2\pi\sigma^2} \right). \label{eq: only ave high-SNR capacity}
	\end{IEEEeqnarray}}
\end{lemma}

\subsubsection{Bounds on Capacity Region}
The inner bound on the capacity region is proposed as follows.
\begin{theorem}[\textbf{Inner Bound}]\label{theorem: inner bound only ave}
	When the input is only subject to the peak-intensity constraint in \eqref{eq:ave cons}, the rate pairs in the set $\cup_{\rho\in[0,1]} \left( R_{1}^{\textnormal{in}}(\rho),R_{2}^{\textnormal{in}}(\rho) \right)$ are all achievable for a two-user OI-BC, where
	{\setlength\abovedisplayskip{4.85pt} 
	\setlength\belowdisplayskip{4.85pt}
	\begin{IEEEeqnarray}{rCl}
		\subnumberinglabel{eq: only ave IB}
		R_{1}^{\textnormal{in}}(\rho) &=& \frac{1}{2} \log \left( 1+\frac{e\rho^2\EE^2}{2 \pi \sigma_1^2} \right),\label{eq: R1 Lb only ave}\\
		R_{2}^{\textnormal{in}}(\rho) &=& \frac{1}{2} \log \left( 1+\frac{e\EE^2}{2 \pi \sigma_2^2} \right) - \mathsf{C}_{\textnormal{ub}} \left(\rho\EE, \sigma_2\right). \label{eq: R2 Lb only ave}
	\end{IEEEeqnarray}}
\end{theorem} 
\begin{proof}
	The proof is similar to that of Theorem~\ref{theorem: inner bound only peak}. Here we mainly emphasize the differences. Different from the peak-intensity constrained OI-BC, here $X_1$ follows $\textsf{Exp}(\rho\EE)$ and $X_2$ follows
	{\setlength\abovedisplayskip{4.85pt} 
	\setlength\belowdisplayskip{4.85pt}
	\begin{IEEEeqnarray}{rCl}
		p_{X_2}(x_2) &=& \rho \delta(s_2) + \left( 1-\rho \right) \times \frac{e^{-\frac{x_2}{\EE}}}{\EE},\quad x_2\in[0,+\infty),
	\end{IEEEeqnarray}}where $\rho\in[0,1]$. Substituting them Lemma~\ref{lemma: decom of exponential}, we obtain that $X$ follows $\textsf{Exp}(\EE)$, which satisfies the average-intensity constraint in \eqref{eq:ave cons}. Besides, note that $X_1$ satisfies $\mathbb{E}[X_1]=\rho\EE$ and is independent of the Gaussian noise $Z_2$, thus we can bound $\mathsf{h}(X_1+Z_2)$ by
	{\setlength\abovedisplayskip{4.85pt} 
	\setlength\belowdisplayskip{4.85pt}
	\begin{IEEEeqnarray}{rCl}
	\mathsf{h}(X_1+Z_2) &\leq& \mathsf{C}_{\textnormal{ub}} ( \rho\EE,\sigma_2) + \mathsf{h}(Z_2).
	\end{IEEEeqnarray}}Finally, following similar steps from \eqref{eq: 21} to \eqref{eq: 69}, we can complete the proof.

\end{proof}

The outer bound on the capacity region is proposed as follows.
\begin{theorem}[\textbf{Outer Bound}]\label{theorem: outer bound only ave}
	When the input is only subject to the average-intensity constraint in \eqref{eq:ave cons}, the capacity region is outer bounded by $\cup_{\rho\in[0,1]} \left( R_{1}^{\textnormal{out}}(\rho),R_{2}^{\textnormal{out}}(\rho) \right)$ for a two-user OI-BC, where
	{\setlength\abovedisplayskip{4.85pt} 
	\setlength\belowdisplayskip{4.85pt}
	\begin{IEEEeqnarray}{rCl}
		\subnumberinglabel{eq: only ave OB}
		R_{1}^{\textnormal{out}}(\rho) &=& \frac{1}{2} \log \left\{1+\frac{\sigma_2^2}{\sigma_1^2}\left(e^{2 \mathsf{C}_{\textnormal{ub}}\left(\rho \EE, \sigma_2\right)}-1\right)\right\},\label{eq: R1 Ub only ave}\\
		R_{2}^{\textnormal{out}}(\rho) &=& \mathsf{C}_{\textnormal{ub}}\left(\EE, \sigma_2\right)-\mathsf{C}_{\textnormal{ub}}\left(\rho \EE, \sigma_2\right).\label{eq: R2 Ub only ave}
	\end{IEEEeqnarray}}
\end{theorem} 
\begin{proof}
	The proof is similar to that of Theorem~\ref{theorem: outer bound only peak}. Note that
	{\setlength\abovedisplayskip{4.85pt} 
	\setlength\belowdisplayskip{4.85pt}
	\begin{IEEEeqnarray}{rCl}
		\mathsf{h}(Y_2|U) &\geq& \mathsf{h}(Z_2),\\
		\mathsf{h}(Y_2|U) &\leq& \mathsf{C}_{\textnormal{ub}}(\EE,\sigma_2) + \mathsf{h}(Z_2). 
	\end{IEEEeqnarray}}Here, $\mathsf{C}_{\textnormal{ub}}(\EE,\sigma_2)$ is monotonically increasing with respect to $\EE$ and approaches zero when $\EE$ tends to $0$. Hence, there exists $\rho\in[0,1]$ such that
	{\setlength\abovedisplayskip{4.85pt} 
	\setlength\belowdisplayskip{4.85pt}
	\begin{IEEEeqnarray}{rCl}
		\mathsf{h}(Y_2|U) = \mathsf{C}_{\textnormal{ub}}(\rho\EE,\sigma_2) + \mathsf{h}(Z_2). 
	\end{IEEEeqnarray}}Finally, following similar steps from to \eqref{eq: 34} to \eqref{eq: 40}, the proof is concluded.

\end{proof}

The high-SNR capacity region is directly obtained by substituting the single-user capacity result in Lemma~\ref{lemma: 1-user capacity only ave} into \eqref{eq: only ave IB} and \eqref{eq: only ave OB}. We summarize it as follows.
\begin{theorem}[\textbf{Asymptotic Capacity Region}]\label{theorem: high SNR capacity region only ave}
	When the input is only subject to the average-intensity constraint in \eqref{eq:ave cons}, at high SNR, the capacity region of a two-user OI-BC asymptotically converges to the region where the rate pair $(R_1,R_2)$ satisfies
	{\setlength\abovedisplayskip{4.85pt} 
	\setlength\belowdisplayskip{4.85pt}
	\begin{IEEEeqnarray}{rCl}
		\subnumberinglabel{eq: only ave high-SNR capacity region}
		R_1 & \ \dot{\leq}\ & \frac{1}{2} \log \left(1+\frac{e\rho^2 \EE^2}{2 \pi \sigma_1^2}\right), \label{eq: high SNR R1 only ave}\\
		R_2 & \ \dot{\leq}\ & \frac{1}{2} \log \left(1+\frac{e\left(1-\rho^2\right) \EE^2}{e\rho^2 \EE^2+2 \pi \sigma_2^2}\right), \label{eq: high SNR R2 only ave}
	\end{IEEEeqnarray}}with $\rho \in [0,1]$.
\end{theorem}

The derived bounds on the capacity region are shown in Fig.~\ref{fig: two-user only ave}, where we assume $\sigma_2 = 2\sigma_1$ and $\mathrm{ESNR}_k=\frac{\EE}{\sigma_k}$, $k\in\{1,2\}$. 
\begin{figure}[!htbp]
	\centering
	\includegraphics[width=3.7in]{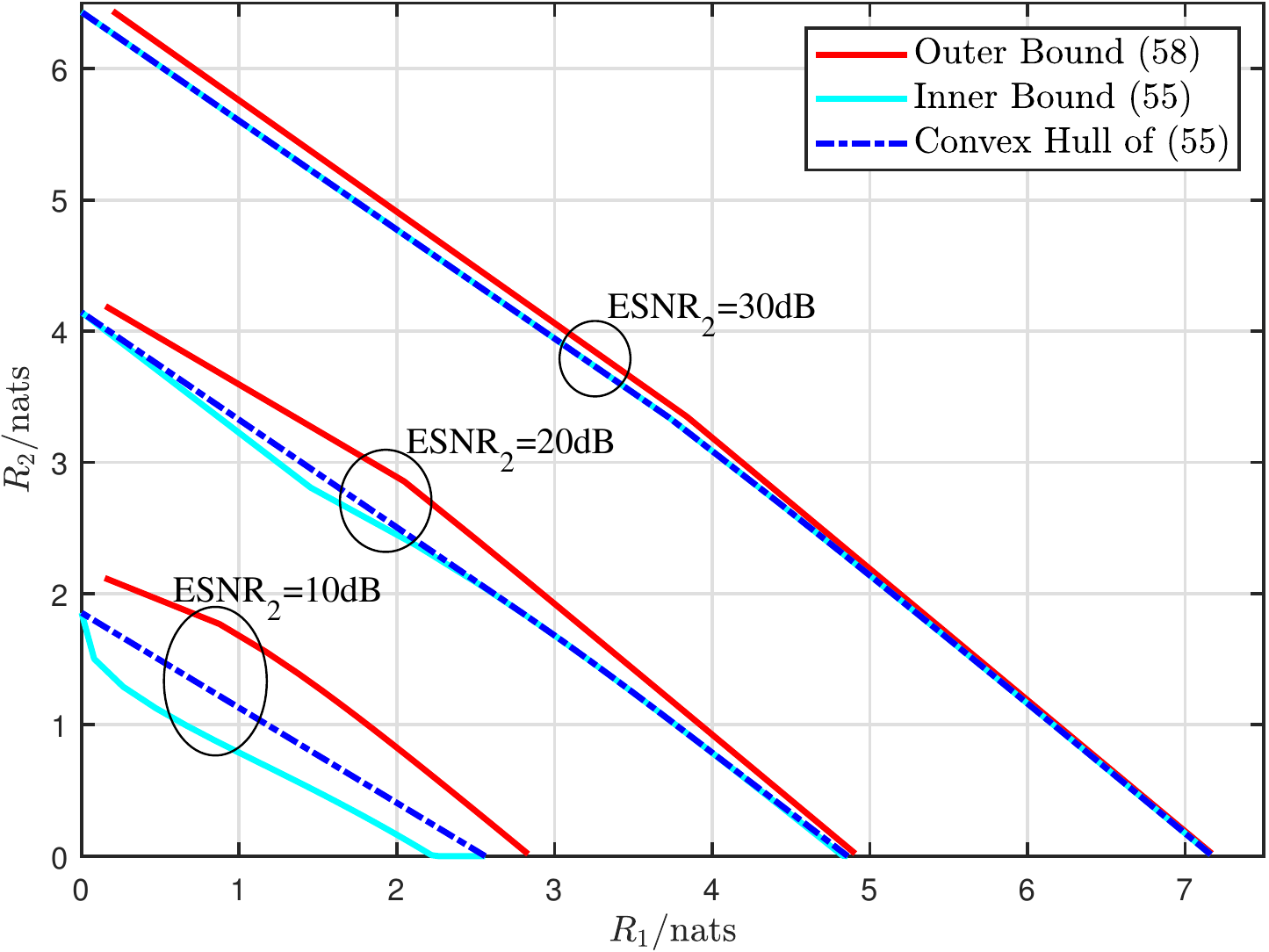}\\
	\caption{Bounds on capacity region of $2$-user OI-BC with average-intensity constraint.}
	\label{fig: two-user only ave}
	\vspace{-0.7cm}
\end{figure}
\subsection{Peak- and Average-Intensity Constrained OI-BC}
\subsubsection{Preliminary}
Considering the peak- and average-intensity constraints in \eqref{eq:peak cons} and \eqref{eq:ave cons}, the decomposition of a truncated exponential random variable and the existing single-user channel capacity is introduced in the following.

\begin{lemma}[\cite{Li2022}]\label{lemma: decom of truncated exponential}
	Given any integer $N\in\mathbb{N}^{+}$, a truncated exponential random variable $S$, i.e.,
	{\setlength\abovedisplayskip{4.85pt} 
		\setlength\belowdisplayskip{4.85pt}
		\begin{IEEEeqnarray}{rCl}
			p_{S}(s) = \frac{1}{\amp} \frac{\mu}{1-e^{-\mu}} e^{ -\frac{\mu s}{\amp} },\quad s\in[0,\amp], \quad \mu>0,
	\end{IEEEeqnarray}}can be decomposed as a sum of two independent random variables $S_1$ and $S_2$, i.e.,
	{\setlength\abovedisplayskip{4.85pt} 
		\setlength\belowdisplayskip{4.85pt}
		\begin{IEEEeqnarray}{rCl}
			p_{S_1}(s_1) &=& \frac{1}{\amp} \frac{\mu}{ 1-e^{ -\frac{\mu}{N} } } e^{ -\frac{\mu s_1}{\amp} },\quad s_1 \in \left[ 0,\frac{\amp}{N} \right],\\
			p_{S_2}(s_2) &=& \frac{ 1-e^{ -\frac{\mu}{N} } }{ 1-e^{-\mu} } \sum_{n=0}^{N-1} e^{-\frac{\mu n}{N}} \delta \left(s_2-\frac{n}{N}\times\amp \right).\label{eq: truncated exp distribution}
	\end{IEEEeqnarray}}	
\end{lemma}For convenience, the distribution in \eqref{eq: truncated exp distribution} is denoted by $\textsf{Texp}(\amp,\mu)$.

\begin{lemma}[\cite{Lapidoth2009}]\label{lemma: 1-user capacity bpower}
	The channel capacity $\mathsf{C}(\amp,\sigma;\mu^\star)$ of a peak- and average-intensity constrained optical intensity channel is upper bounded by
	{\setlength\abovedisplayskip{4.85pt} 
	\setlength\belowdisplayskip{4.85pt}
	\begin{IEEEeqnarray}{rCl}
		\mathsf{C}_{\textnormal{ub}}(\amp,\sigma;\mu^\star) = \inf_{\beta_1>0,\beta_2>0} B(\beta_1,\beta_2),\label{eq: 1-user ub capacity bpower}
	\end{IEEEeqnarray}}where 
	{\setlength\abovedisplayskip{4.85pt} 
	\setlength\belowdisplayskip{4.85pt}
	\begin{IEEEeqnarray}{rCl}
		\beta_1 &=& \sigma \log\left( 1 + \frac{\amp}{\sigma}\right)\\
		\beta_2 &=& \mu^\star \Biggl( 1 - \exp\biggl(- \frac{\beta_1^2}{2\sigma^2} \times \Bigl(\frac{1}{\mu^\star} - \frac{ e^{-\mu^\star} }{ 1-e^{-\mu^\star} }\Bigr) \biggr) \Biggr),
	\end{IEEEeqnarray}}and $B(\beta_1,\beta_2)$ is given in \eqref{eq: 80a} at the bottom of the next page. The parameter $\mu^\star$ is the unique solution to the following equation
	{\setlength\abovedisplayskip{4.85pt} 
	\setlength\belowdisplayskip{4.85pt}
	\begin{IEEEeqnarray}{rCl}
		\frac{1}{\mu^\star} - \frac{ e^{-\mu^\star} }{ 1-e^{-\mu^\star} } = \alpha.
	\end{IEEEeqnarray}}At high SNR,
	{\setlength\abovedisplayskip{4.85pt} 
	\setlength\belowdisplayskip{4.85pt}
	\begin{IEEEeqnarray}{rCl}
		\mathsf{C}_{\textnormal{ub}}(\amp,\sigma;\mu^\star) \doteq \frac{1}{2}\log 
		\left(1+ \exp\left( 2 - \frac{ 2\mu^\star e^{-\mu^\star} }{ 1-e^{-\mu^\star} } \right) \left( \frac{1- e^{-\mu^\star} }{\mu^\star}
		\right)^2 \frac{ \amp^2 }{2\pi e\sigma^2} \right).\label{eq: peak and ave high-SNR capacity}
	\end{IEEEeqnarray}}
\end{lemma}

\begin{figure*}[!b]
	\normalsize
	\setcounter{MYtempeqncnt}{\value{equation}}
	\setcounter{equation}{83}
	\vspace*{4pt}
	\hrulefill
	{\setlength\abovedisplayskip{4.85pt} 
	\setlength\belowdisplayskip{4.85pt}
	\begin{align}
		B(\beta_1,\beta_2) =&  \left(1-\mathcal{Q}\left(\frac{\beta_1+\left( \frac{1}{\mu^\star} - \frac{ e^{-\mu^\star} }{ 1-e^{-\mu^\star} } \right)  \amp}{\sigma}\right)-\mathcal{Q}\left(\frac{\beta_1+(1-\left( \frac{1}{\mu^\star} - \frac{ e^{-\mu^\star} }{ 1-e^{-\mu^\star} } \right) ) \amp}{\sigma}\right)\right) \nonumber\\
		& \times \log \left(\frac{\amp}{\sigma} \cdot \frac{e^{\frac{\beta_2 \beta_1}{\amp}}-e^{-\beta_2\left(1+\frac{\beta_1}{\amp}\right)}}{\sqrt{2 \pi} \beta_2\left(1-2 \mathcal{Q}\left(\frac{\beta_1}{\sigma}\right)\right)}\right) \nonumber\\
		& -\frac{1}{2}+\mathcal{Q}\left(\frac{\beta_1}{\sigma}\right)+\frac{\beta_1}{\sqrt{2 \pi} \sigma} e^{-\frac{\beta_1^2}{2 \sigma^2}} +\frac{\sigma}{\amp} \frac{\beta_2}{\sqrt{2 \pi}}\left(e^{-\frac{\beta_1^2}{2 \sigma^2}}-e^{-\frac{(\amp+\beta_1)^2}{2 \sigma^2}}\right) \nonumber\\
		& +\beta_2 \left( \frac{1}{\mu^\star} - \frac{ e^{-\mu^\star} }{ 1-e^{-\mu^\star} } \right) \left(1-2 \mathcal{Q}\left(\frac{\beta_1+\frac{\amp}{2}}{\sigma}\right)\right).\tag{\ref{eq: 1-user ub capacity bpower}a} \label{eq: 80a}
	\end{align}}
	\setcounter{equation}{\value{MYtempeqncnt}}
\end{figure*}

\subsubsection{Bounds on Capacity Region}
The inner bound on the capacity region is proposed as follows.
\begin{theorem}[\textbf{Inner Bound}]\label{theorem: inner bound peak and ave}
	When the input is subject to both peak- and average-intensity constraint in \eqref{eq:peak cons} and \eqref{eq:ave cons}, the rate pairs in the set $\mathsf{Conv}\{ \cup_{N\in\mathbb{N}^{+}} \left( R_{1}^{\textnormal{in}}(N),R_{2}^{\textnormal{in}}(N) \right) \}$ are all achievable for a two-user OI-BC, where
	{\setlength\abovedisplayskip{4.85pt} 
	\setlength\belowdisplayskip{4.85pt}
	\begin{IEEEeqnarray}{rCl}
		\subnumberinglabel{eq: peak and ave IB}
		R_{1}^{\textnormal{in}}(N) 
		&=& \frac{1}{2}\log \left( 1 + 
		\exp\left(2 - \frac{ 2\mu^\star  e^{-\frac{\mu^\star}{N}}}{ N\left(1-e^{-\frac{\mu^\star}{N}} \right) }\right) \left( \frac{1- e^{-\frac{\mu^\star}{N}} }{\mu^\star}
		\right)^2 \frac{ \amp^2 }{2\pi e \sigma_1^2} \right),\label{eq: R1 Lb peak and ave}\\
		R_{2}^{\textnormal{in}}(N) 
		&=& \frac{1}{2}\log \left(1+ \exp\left(2 - \frac{2\mu^\star e^{-\mu^\star}}{1-e^{-\mu^\star}}\right) \left( \frac{1- e^{-\mu^\star} }{\mu^\star}
		\right)^2 \frac{ \amp^2 }{2\pi e\sigma_2^2} \right)- \mathsf{C}_{\textnormal{ub}} \left(\frac{\amp}{N}, \sigma_2; \frac{\mu^\star}{N} \right). 
		\label{eq: R2 Lb peak and ave}
	\end{IEEEeqnarray}}
	\end{theorem} 
\begin{proof}
	Assume $X_1$ follows $\textsf{Texp}(\frac{\amp}{N},\frac{\mu^\star}{N})$
	and $X_2$ follows
	{\setlength\abovedisplayskip{4.85pt} 
	\setlength\belowdisplayskip{4.85pt}
	\begin{IEEEeqnarray}{rCl}
		p_{X_2}(x_2) &=& \frac{ 1-e^{ -\frac{\mu^\star}{N} } }{ 1-e^{-\mu^\star} } \sum_{n=0}^{N-1} e^{-\frac{\mu^\star n}{N}} \delta \left(x_2-\frac{n}{N}\times\amp \right),
	\end{IEEEeqnarray}}where $N\in\mathbb{N}^{+}$. Substituting them into Lemma~\ref{lemma: decom of truncated exponential}, we obtain that $X$ follows $\textsf{Texp}(\amp,\mu^\star)$, which satisfies the peak- and average-intensity constraints in \eqref{eq:peak cons} and \eqref{eq:ave cons}. Furthermore, note that $X_1$ satisfies $X_1\in[0,{\amp}/{N}$] and $\mathbb{E}[X_1]=\rho\EE$, and is independent of the Gaussian noise $Z_2$, thus we can bound $\mathsf{h}(X_1+Z_2)$ by
	{\setlength\abovedisplayskip{4.85pt} 
	\setlength\belowdisplayskip{4.85pt}
	\begin{IEEEeqnarray}{rCl}
		\mathsf{h}(X_1+Z_2) \leq \mathsf{C}_{\textnormal{ub}} \left( \frac{\amp}{N},\sigma_2;\frac{\mu^\star}{N} \right) + \mathsf{h}(Z_2).
	\end{IEEEeqnarray}}Finally, following similar steps from \eqref{eq: 21} to \eqref{eq: 69}, we can complete the proof.

\end{proof}

The following outer bound has been given in \cite[Theorem 1]{Chaaban2016-2}. Here, we can also utilize the conditional EPI to provide a new proof, which is similar to the proof of Theorem~\ref{theorem: outer bound only peak} and omitted here.
\begin{lemma}[\textbf{Outer Bound}]\label{lemma: outer bound peak and ave}
	When the input is subject to both peak- and average-intensity constraints in \eqref{eq:peak cons} and \eqref{eq:ave cons}, the capacity region is outer bounded by $\cup_{\rho\in[0,1]} \left( R_{1}^{\textnormal{out}}(\rho),R_{2}^{\textnormal{out}}(\rho) \right)$ for a two-user OI-BC, where
	{\setlength\abovedisplayskip{4.85pt} 
	\setlength\belowdisplayskip{4.85pt}
	\begin{IEEEeqnarray}{rCl}
		\subnumberinglabel{eq: both power OB}
		R_{1}^{\textnormal{out}}(\rho) &=& \frac{1}{2} \log \left\{1+\frac{\sigma_2^2}{\sigma_1^2}\left(e^{2 \mathsf{C}_{\textnormal{ub}}\left(\rho\amp, \sigma_2;\mu^\star \right)}-1\right)\right\},\label{eq: R1 Ub peak and ave}\\
		R_{2}^{\textnormal{out}}(\rho) &=& \mathsf{C}_{\textnormal{ub}}\left(\amp, \sigma_2;\mu^\star\right)-\mathsf{C}_{\textnormal{ub}}\left(\rho \amp, \sigma_2;\mu^\star\right),\label{eq: R2 Ub peak and ave}
	\end{IEEEeqnarray}}
\end{lemma}

The high-SNR capacity region is proposed as follows.
\begin{theorem}[\textbf{Asymptotic Capacity Region}]\label{theorem: high SNR capacity region peak and ave}
	When the input is subject to both peak- and average-intensity constraints in \eqref{eq:peak cons} and \eqref{eq:ave cons}, at high SNR, the capacity region of a two-user OI-BC asymptotically converges to the region where the rate pair $(R_1,R_2)$ satisfies
	{\setlength\abovedisplayskip{4.85pt} 
	\setlength\belowdisplayskip{4.85pt}
	\begin{IEEEeqnarray}{rCl}
		\subnumberinglabel{eq: peak and ave high-SNR capacity region}
		R_1 & \ \dot{\leq}\ & \frac{1}{2} \log \left(1+ \exp \left(2 - \frac{2\rho\mu^\star e^{-\rho\mu^\star}}{1-e^{-\rho\mu^\star}}\right) \left( \frac{1- e^{-\rho\mu^\star} }{\mu^\star}
		\right)^2 \frac{ \amp^2 }{2\pi e \sigma_1^2} \right), \label{eq: high SNR R1 peak and ave}\\
		R_2 & \ \dot{\leq}\ & \frac{1}{2}\log \left(\frac{ \exp  \left(2 - \frac{2\mu^\star e^{-\mu^\star}}{1-e^{-\mu^\star}}\right) \left( \frac{1- e^{-\mu^\star} }{\mu^\star}
		\right)^2 \amp^2 + 2\pi e\sigma_2^2 }{ \exp \left(2 - \frac{2\rho\mu^\star e^{-\rho\mu^\star}}{1-e^{-\rho\mu^\star}}\right) \left( \frac{1- e^{-\rho\mu^\star} }{\mu^\star}
		\right)^2 \amp^2 + 2\pi e\sigma_2^2 } \right), \label{eq: high SNR R2 peak and ave}
	\end{IEEEeqnarray}}with $\rho \in [0,1]$.
\end{theorem}
\begin{proof}
	The proof follows similar arguments as the proof of Theorem~\ref{theorem: high SNR capacity region only peak}. Before that, we need to derive a new outer bound which is valid in the high SNR regime. Note that $\mathsf{h}(Y_2|U)$ can be bounded by
	{\setlength\abovedisplayskip{4.85pt} 
	\setlength\belowdisplayskip{4.85pt}
	\begin{IEEEeqnarray}{rCl}
		\mathsf{h}(Y_2|U) &\geq& \mathsf{h}(Z_2),\label{eq: 79}\\
		\mathsf{h}(Y_2|U) &\leq& \mathsf{C}_{\textnormal{ub}}(\amp,\sigma_2;\mu^\star) + \mathsf{h}(Z_2). \label{eq: 80}
	\end{IEEEeqnarray}}The function $\mathsf{C}_{\textnormal{ub}}(\rho\amp,\sigma_2;\rho\mu^\star)$, $\rho\in[0,1]$, is monotonically increasing with respect to $\rho$ and approaches zeros when $\rho$ tends to $0$ at high SNR. (More details about the monotonicity of $\mathsf{C}_{\textnormal{ub}}(\rho\amp,\sigma_2;\rho\mu^\star)$ can been found in Appendix~\ref{app: monotonicity of C_mu(A,sigma)}). Combined with \eqref{eq: 79} and \eqref{eq: 80}, there exists $\rho\in[0,1]$ such that
	{\setlength\abovedisplayskip{4.85pt} 
	\setlength\belowdisplayskip{4.85pt}
	\begin{IEEEeqnarray}{rCl}
		\mathsf{h}(Y_2|U) = \mathsf{C}_{\textnormal{ub}}(\rho\amp,\sigma_2;\rho\mu^\star) + \mathsf{h}(Z_2). 
	\end{IEEEeqnarray}}Following similar arguments from \eqref{eq: 34} to \eqref{eq: 40}, we obtain that at high SNR,
	{\setlength\abovedisplayskip{4.85pt} 
	\setlength\belowdisplayskip{4.85pt}
	\begin{IEEEeqnarray}{rCl}
	R_1 
	&\ \dot{\leq}\ & \frac{1}{2} \log \left(1+ \exp \left(2 - \frac{2\rho\mu^\star e^{-\rho\mu^\star}}{1-e^{-\rho\mu^\star}}\right) \left( \frac{1- e^{-\rho\mu^\star} }{\mu^\star}
	\right)^2 \frac{ \amp^2 }{2\pi e \sigma_1^2} \right),\\
	R_2 
	&\ \dot{\leq}\ &  \frac{1}{2}\log \left(\frac{ \exp  \left(2 - \frac{2\mu^\star e^{-\mu^\star}}{1-e^{-\mu^\star}}\right) \left( \frac{1- e^{-\mu^\star} }{\mu^\star}
		\right)^2 \amp^2 + 2\pi e\sigma_2^2 }{ \exp \left(2 - \frac{2\rho\mu^\star e^{-\rho\mu^\star}}{1-e^{-\rho\mu^\star}}\right) \left( \frac{1- e^{-\rho\mu^\star} }{\mu^\star}
		\right)^2 \amp^2 + 2\pi e\sigma_2^2 } \right). \label{eq: 120}
	\end{IEEEeqnarray}}Finally, combining the above newly derived outer bound and the inner bound in Theorem~\ref{theorem: inner bound peak and ave}, we can complete the proof.
	
\end{proof}

The derived bounds on the capacity region are shown in Fig.~\ref{fig: two-user bpower}, where we assume $\sigma_2 = 2\sigma_1$, $\alpha=0.4$, and $\mathrm{ASNR}_k=\frac{\amp}{\sigma_k}$, $k\in\{1,2\}$. 

\begin{figure}[!htbp]
	\centering
	\includegraphics[width=3.7in]{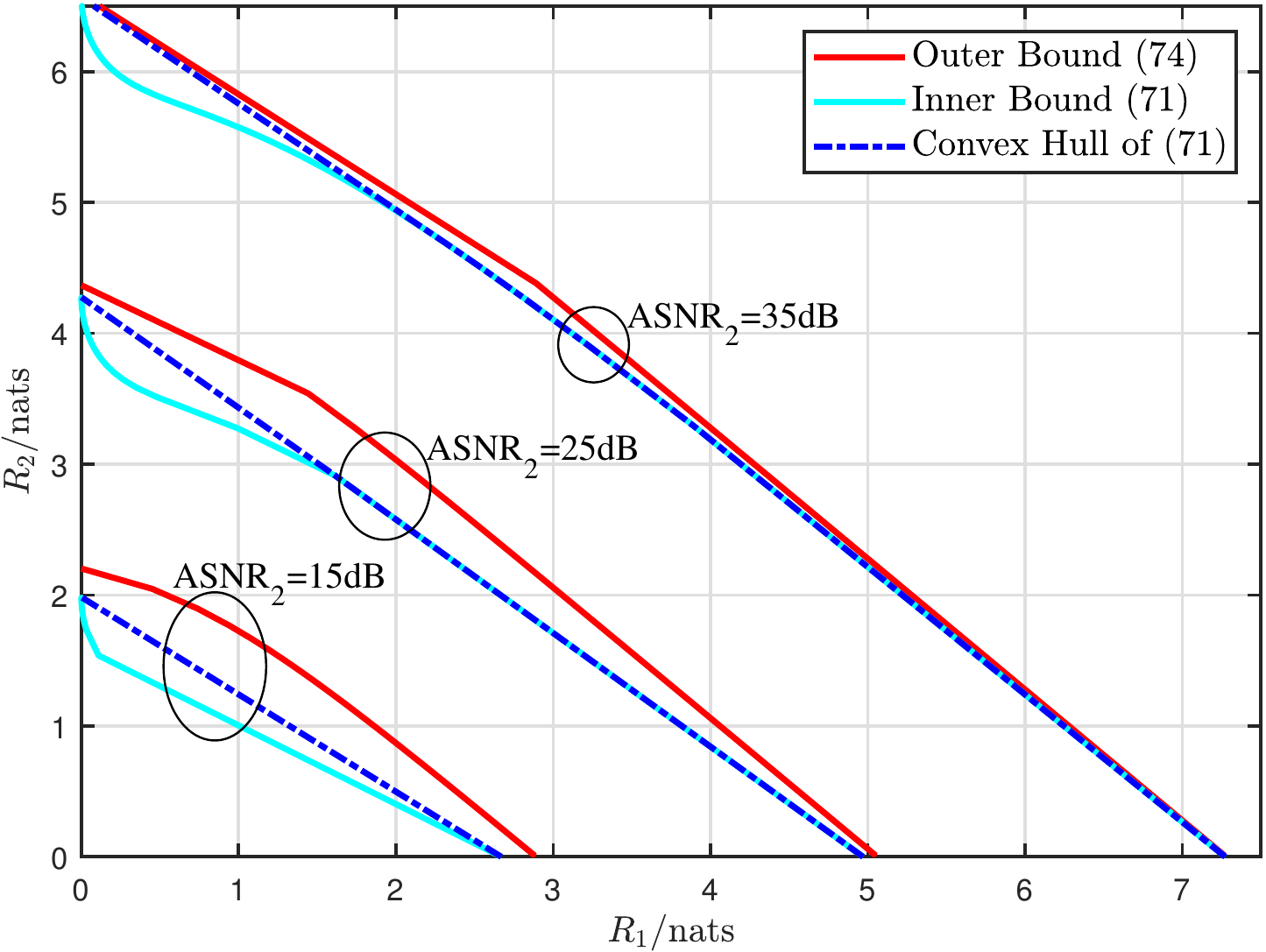}\\
	\caption{Bounds on capacity region of $2$-user OI-BC with peak- and average-intensity constraint when $\alpha=0.4$.}
	\label{fig: two-user bpower}
	\vspace{-0.7cm}
\end{figure}

\section{Capacity Region Characterization of $K$-User OI-BCs}\label{sec: K-user capacity region}
In this section, we consider three input constraints and extend the capacity regions of two-user OI-BCs to $K$-user OI-BCs. In each case of input constraints, the inner bound is derived first and then the outer bound. Finally, the high-SNR capacity region is characterized based on these bounds.

\subsection{Peak-Intensity Constrained OI-BC}
The inner bound on the capacity region is derived based on the SC scheme.
\begin{theorem}[\textbf{Inner Bound}]\label{theorem: K-user inner bound only peak}
	When the input is only subject to the peak-intensity constraint in \eqref{eq:peak cons}, the rate tuples in the set $\mathsf{Conv}\{ \cup_{\bm{N}\in\mathfrak{D}_{\bm{N}}} \left( R_{1}^{\textnormal{in}}(\bm{N}),R_{2}^{\textnormal{in}}(\bm{N}),\cdots,R_{K}^{\textnormal{in}}(\bm{N}) \right) \}$ are all achievable for a $K$-user OI-BC, where
	{\setlength\abovedisplayskip{4.85pt} 
	\setlength\belowdisplayskip{4.85pt}
	\begin{IEEEeqnarray}{rCl}
		R_{k}^{\textnormal{in}}(\bm{N}) &=& \frac{1}{2} \log \left( 1+\frac{\amp^2}{2 \pi e (\prod_{n=k}^K {N}_{n})^2\sigma_k^2} \right) -\mathcal{C}_{\textnormal{up}}\left( \frac{\amp}{\prod_{n=k-1}^K {N}_{n}},\sigma_k\right), \ k\in[K], \qquad\label{eq: Rk Lb only peak K-user}
	\end{IEEEeqnarray}}with $\mathfrak{D}_{\bm{N}}=\{ \bm{N}=[N_0,N_1,\cdots,N_K]: {N}_{n}\in \mathbb{N}^{+}, \ \forall n\in \{0\cup[K]\}, \ {N}_0\gg 0,\ {N}_K=1 \}$.\footnote{We assume that $N_0\gg0$ is equivalent to $\frac{1}{{N}_0}=0$.}
\end{theorem} 
\begin{proof}	
	We first encode the messages $M_1$, $M_2$, $\cdots$, $M_K$ independently into signals $X_1$, $X_2$, $\cdots$, $X_K$, where $X_1$ follows $\textsf{U} \left[0,\frac{\amp}{ \prod_{n=1}^K {N}_n} \right]$
	and $\forall k\in\{2,3,\cdots,K\}$:
	{\setlength\abovedisplayskip{4.85pt} 
	\setlength\belowdisplayskip{4.85pt}
	\begin{IEEEeqnarray}{rCl}
		p_{X_k}(x_k) = \frac{1}{{N}_{k-1}} \sum_{i=0}^{{N}_{k-1}-1} \delta \left( x_k - \frac{i}{{N}_{k-1}} \times \frac{ \amp }{ \prod_{n=k}^K {N}_n } \right).
	\end{IEEEeqnarray}}Then we adopt an SC scheme such that $X=\sum_{n=1}^K X_n$. Denote 
	{\setlength\abovedisplayskip{4.85pt} 
	\setlength\belowdisplayskip{4.85pt}
	\begin{IEEEeqnarray}{rCl}
		X_{k,\textnormal{sum}}=\sum_{n=1}^k X_n,\quad k\in[K].\label{eq: X_sum def}
	\end{IEEEeqnarray}}Combined with Lemma~\ref{lemma: decom of uniform}, we obtain that $X_{k,\textnormal{sum}}$ follows $\textsf{U} \left[ 0, \frac{\amp}{\prod_{n=k}^K {N}_n} \right]$ and $X$ follows $\textsf{U}$$[0,\amp]$, which satisfies the peak-intensity constraint in \eqref{eq:peak cons}.
	In Lemma \ref{lemma: Kuser}, we instantiate $U_k$ into $U_k=\sum_{n=k}^K X_n$. Therefore, we can compute the achievable rate $\mathsf{I}\bigl(\sum_{n=k}^K X_n;Y_k \bigm| \sum_{n=k+1}^K X_n\bigr)$ to be as the inner bounds on $R_k$, $k\in[K]$. To do it, we simplify $\mathsf{I}\bigl(\sum_{n=k}^K X_n;Y_k \bigm| \sum_{n=k+1}^K X_n\bigr)$ as
	{\setlength\abovedisplayskip{4.85pt} 
	\setlength\belowdisplayskip{4.85pt}
	\begin{IEEEeqnarray}{rCl}
	\mathsf{I}\biggl(\sum_{n=k}^K X_n;Y_k \biggm| \sum_{n=k+1}^K X_n\biggr) 	
		&=& \mathsf{I} \left( X_k; \sum\nolimits_{n=1}^k X_n +Z_k \right)\\
		&=& \mathsf{h} \left( \sum\nolimits_{n=1}^k X_n +Z_k \right) - \mathsf{h} \left( \sum\nolimits_{n=1}^{k-1} X_n +Z_k \right) \\
		&\geq& \frac{1}{2} \log \left( e^{ 2\mathsf{h} \left(\sum\nolimits_{n=1}^k X_n \right) } + e^{ 2 \mathsf{h}(Z_k)} \right) - \mathsf{h} \left( \sum\nolimits_{n=1}^{k-1} X_n +Z_k \right) \label{eq: 150}\\
		&\geq& \frac{1}{2} \log \left( 1+\frac{\amp^2}{2 \pi e (\prod_{n=k}^K {N}_{n})^2\sigma_k^2} \right) -\mathsf{C}_{\textnormal{ub}}\left( \frac{\amp}{\prod_{n=k-1}^K {N}_{n}},\sigma_k\right), \quad\label{eq: I(Uk;Yk|U_{k+1}) inequality}
	\end{IEEEeqnarray}}where \eqref{eq: 150} follows from the EPI; and \eqref{eq: I(Uk;Yk|U_{k+1}) inequality} follows from the fact that $\sum\nolimits_{n=1}^{k-1} X_n$ is limited in $\left[ \frac{\amp}{\prod_{n=k-1}^K {N}_{n}} \right]$ and independent of the Gaussian noise $Z_k$. Combined with \eqref{eq: I(Uk;Yk|U_{k+1}) inequality}, the proof is concluded.
\end{proof}

The outer bound on the capacity region is derived based on the conditional EPI.
\begin{theorem}[\textbf{Outer Bound}]\label{theorem: K-user outer bound only peak}
	When the input is only subject to the peak-intensity constraint in \eqref{eq:peak cons}, the capacity region for a $K$-user OI-BC is outer bounded by\\
	$\cup_{\bm{\rho}\in\mathfrak{D}_{\bm{\rho}}} \left( R_{1}^{\textnormal{out}}(\bm{\rho}),R_{2}^{\textnormal{out}}(\bm{\rho}),\cdots,R_{K}^{\textnormal{out}}(\bm{\rho}) \right)$, where
	{\setlength\abovedisplayskip{4.85pt} 
	\setlength\belowdisplayskip{4.85pt}
	\begin{IEEEeqnarray}{rCl}
		R_{k}^{\textnormal{out}}(\bm{\rho}) &=& \frac{1}{2} \log \left\{\frac{ \sigma_k^2 + \sigma_K^2 \left( e^{ 2\mathsf{C}_{\textnormal{ub}}({\rho}_k\amp,\sigma_K) } -1 \right) }{ \sigma_k^2 + \sigma_K^2 \left( e^{ 2\mathsf{C}_{\textnormal{ub}}({\rho}_{k-1}\amp,\sigma_K) } -1 \right) } \right\},\quad k\in[K],\label{eq: Ub only peak K-user}
	\end{IEEEeqnarray}}with $\mathfrak{D}_{\bm{\rho}}=\left\{\bm{\rho}=[{\rho}_0,{\rho}_1,\cdots,{\rho}_K]: {\rho}_k \in[0,1], \ {\rho}_{k-1}\leq {\rho}_{k}, \ \forall k\in[K], \ {\rho}_0=0,\ {\rho}_K=1\right\}$.
\end{theorem} 
\begin{proof}
	We resort to Lemma \ref{lemma: Kuser} to derive the upper bound on the achievable rate $R_K$ of user $K$ first and then the achievable rates $R_k$s of the rest of users, $k\in[K-1]$. When $U_K \rightarrow X \rightarrow Y_K$ forms a Markov chain, it follows that
	{\setlength\abovedisplayskip{4.85pt} 
	\setlength\belowdisplayskip{4.85pt}
	\begin{IEEEeqnarray}{rCl}
		\mathsf{h}(Y_K|U_K)
		&\geq& \mathsf{h}(Y_K|X) 
		= \mathsf{h}(Z_K).
	\end{IEEEeqnarray}}Furthermore, we have
	{\setlength\abovedisplayskip{4.85pt} 
	\setlength\belowdisplayskip{4.85pt}
	\begin{IEEEeqnarray}{rCl}
		\mathsf{h}(Y_K|U_K) &\leq& \mathsf{h}(Y_K) \leq \mathsf{C}_{\textnormal{ub}}(\amp,\sigma_K) + \mathsf{h}(Z_K). \label{eq: h(YK|UK) UB}
	\end{IEEEeqnarray}}Hence, there exists ${\rho}_{K-1}\in[0,1]$ such that
	{\setlength\abovedisplayskip{4.85pt} 
	\setlength\belowdisplayskip{4.85pt}
	\begin{IEEEeqnarray}{rCl}
		\mathsf{h}(Y_K|U_K) &=& \mathsf{C}_{\textnormal{ub}}({\rho}_{K-1}\amp,\sigma_K) + \mathsf{h}(Z_K)\\
		&=& \frac{1}{2} \log \left( 2\pi e\sigma_K^2 + 2\pi e\sigma_K^2 ( e^{ 2 \mathsf{C}_{\textnormal{ub}}({\rho}_{K-1}\amp,\sigma_K) } -1 ) \right). \label{eq: h(YK|UK) equality}
	\end{IEEEeqnarray}}By Lemma \ref{lemma: Kuser}, the achievable rate of user $K$ can be upper bounded by
	{\setlength\abovedisplayskip{4.85pt} 
	\setlength\belowdisplayskip{4.85pt}
	\begin{IEEEeqnarray}{rCl}
		R_K &\leq& \mathsf{I}(U_K;Y_K)\\
		&=& \mathsf{h}(Y_K) - \mathsf{h}(Y_K|U_K)\\
		&\leq& \mathsf{C}_{\textnormal{ub}}(\amp,\sigma_K) + \frac{1}{2} \log \left( 2\pi e\sigma_K^2 \right) -\frac{1}{2} \log \left( 2\pi e\sigma_K^2 + 2\pi e\sigma_K^2 ( e^{ 2 \mathsf{C}_{\textnormal{ub}}({\rho}_{K-1}\amp,\sigma_K) } -1 ) \right)\\
		&=& \frac{1}{2} \log \left\{\frac{ \sigma_K^2 + \sigma_K^2 \left( e^{ 2\mathsf{C}_{\textnormal{ub}}({\rho}_K\amp,\sigma_K) } -1 \right) }{ \sigma_K^2 + \sigma_K^2 \left( e^{ 2\mathsf{C}_{\textnormal{ub}}({\rho}_{K-1}\amp,\sigma_K) } -1 \right) } \right\}. \label{eq: RK inequality2}
	\end{IEEEeqnarray}}Besides, the achievable rate $R_k$ of user $k$, $k\in[K-1]$, can be upper bounded by
	{\setlength\abovedisplayskip{4.85pt} 
	\setlength\belowdisplayskip{4.85pt}
	\begin{IEEEeqnarray}{rCl}
		R_k &\leq& \mathsf{I}(U_k;Y_k|U_{k+1})\\
		&=& \mathsf{h}(Y_k|U_{k+1}) - \mathsf{h}(Y_k|U_{k}),\quad k\in[K-1],\label{eq: Rk inequality}
	\end{IEEEeqnarray}}where \eqref{eq: Rk inequality} holds since $U_{k+1}\rightarrow U_{k}\rightarrow Y_{k}$ forms a Markov chain. To derive an upper bound on \eqref{eq: Rk inequality}, we first analyze $\mathsf{h}(Y_k|U_{k})$ and then $\mathsf{h}(Y_k|U_{k+1})$.

To analyze $\mathsf{h}(Y_k|U_k)$, we assume $\forall k\in[K]$,
	{\setlength\abovedisplayskip{4.85pt} 
	\setlength\belowdisplayskip{4.85pt}
	\begin{IEEEeqnarray}{rCl}
		\mathsf{h}(Y_k|U_k) = \frac{1}{2} \log \left( 2\pi e\sigma_k^2 + 2\pi e\sigma_K^2 ( e^{ 2 \mathsf{C}_{\textnormal{ub}}({\rho}_{k-1}\amp,\sigma_K) } -1 ) \right). \label{eq: h(Yk|Uk) equality}
	\end{IEEEeqnarray}}From \eqref{eq: h(YK|UK) equality}, we find that \eqref{eq: h(Yk|Uk) equality} is true if $k=K$. Then, we fix a particular $i\in\{2,\cdots,K\}$ and assume \eqref{eq: h(Yk|Uk) equality} is true if $k=i$, i.e.,
	{\setlength\abovedisplayskip{4.85pt} 
	\setlength\belowdisplayskip{4.85pt}
	\begin{IEEEeqnarray}{rCl}
		\mathsf{h}(Y_i|U_i) = \frac{1}{2} \log \left( 2\pi e\sigma_i^2 + 2\pi e\sigma_K^2 ( e^{ 2 \mathsf{C}_{\textnormal{ub}}({\rho}_{i-1}\amp,\sigma_K) } -1 ) \right). \label{eq: h(Yi|Ui) equality}
	\end{IEEEeqnarray}}We upper bound $\mathsf{h}(Y_{i-1}|U_{i-1})$ by
	{\setlength\abovedisplayskip{4.85pt} 
	\setlength\belowdisplayskip{4.85pt}
	\begin{IEEEeqnarray}{rCl}
		\mathsf{h}(Y_{i-1}|U_{i-1})
		 &\leq& \mathsf{h}(Y_{i-1}|U_i) \label{eq: degraded channel inequaliy}\\
		&\leq& \frac{1}{2} \log \left( e^{2\mathsf{h}(Y_i|U_i)} - 2\pi e ( \sigma_i^2-\sigma_{i-1}^2 )\right) \label{eq: EPI2}\\
		&=& \frac{1}{2} \log \left( 2\pi e\sigma_{i-1}^2 + 2\pi e\sigma_K^2 ( e^{ 2 \mathsf{C}_{\textnormal{ub}}({\rho}_{i-1}\amp,\sigma_K) } -1 ) \right),\label{eq: h(Y{i-1}|U_{i-1}) equality}
	\end{IEEEeqnarray}}where \eqref{eq: degraded channel inequaliy} holds since $U_i\rightarrow U_{i-1}\rightarrow Y_{i-1}$ forms a Markov chain; \eqref{eq: EPI2} is obtained by substituting the conditional EPI for $Y_{i-1}+\widetilde{Z}_i = Y_i$; and \eqref{eq: h(Y{i-1}|U_{i-1}) equality} is obtained by substituting \eqref{eq: h(Yi|Ui) equality} into \eqref{eq: EPI2}. On the other hand, since $U_i\rightarrow X\rightarrow Y_{i-1}$ also forms a Markov chain, we lower bound $\mathsf{h}(Y_{i-1}|U_{i-1})$ by
	{\setlength\abovedisplayskip{4.85pt} 
	\setlength\belowdisplayskip{4.85pt}
	\begin{IEEEeqnarray}{rCl}
		\mathsf{h}(Y_{i-1}|U_{i-1}) \geq \mathsf{h}(Y_{i-1}|X) = \frac{1}{2} \log(2\pi e\sigma_{i-1}^2). \label{eq: 139}
	\end{IEEEeqnarray}}With \eqref{eq: h(Y{i-1}|U_{i-1}) equality} and \eqref{eq: 139}, there exists ${\rho}_{i-2}\in[0,{\rho}_{i-1}]$ such that 
	{\setlength\abovedisplayskip{4.85pt} 
	\setlength\belowdisplayskip{4.85pt}
	\begin{IEEEeqnarray}{rCl}
		\mathsf{h}(Y_{i-1}|U_{i-1}) = \frac{1}{2} \log \left( 2\pi e\sigma_{i-1}^2 + 2\pi e\sigma_K^2 ( e^{ 2 \mathsf{C}_{\textnormal{ub}}({\rho}_{i-2}\amp,\sigma_K) } -1 ) \right). \label{eq: h(Y{i-1}|U{i-1}) equality}
	\end{IEEEeqnarray}}As a result, \eqref{eq: h(Yk|Uk) equality} is also true if $k=i-1$. By mathematical induction, \eqref{eq: h(Yk|Uk) equality} is true for $\forall k\in[K]$. To analyze $\mathsf{h}(Y_k|U_{k+1})$, we first note that $Y_k+\widetilde{Z}_{k+1}=Y_{k+1}$. By the conditional EPI, we obtain that
	{\setlength\abovedisplayskip{4.85pt} 
	\setlength\belowdisplayskip{4.85pt}
	\begin{IEEEeqnarray}{rCl}
		\mathsf{h}(Y_k|U_{k+1}) 
		&\leq& \frac{1}{2} \log \left( e^{2\mathsf{h}(Y_{k+1}|U_{k+1})} - 2\pi e ( \sigma_{k+1}^2-\sigma_{k}^2 )\right) \\
		&=& \frac{1}{2} \log \left( 2\pi e\sigma_{k}^2 + 2\pi e\sigma_K^2 ( e^{ 2 \mathsf{C}_{\textnormal{ub}}({\rho}_{k}\amp,\sigma_K) } -1 ) \right), \ k\in[K-1]. \label{eq: EPI3}
	\end{IEEEeqnarray}}

	Substituting \eqref{eq: h(Yk|Uk) equality} and \eqref{eq: EPI3} into \eqref{eq: Rk inequality}, we have
	{\setlength\abovedisplayskip{4.85pt} 
	\setlength\belowdisplayskip{4.85pt}
	\begin{IEEEeqnarray}{rCl}
	R_k &\leq& \frac{1}{2} \log \left\{\frac{ \sigma_k^2 + \sigma_K^2 \left( e^{ 2\mathsf{C}_{\textnormal{ub}}({\rho}_k\amp,\sigma_K) } -1 \right) }{ \sigma_k^2 + \sigma_K^2 \left( e^{ 2\mathsf{C}_{\textnormal{ub}}({\rho}_{k-1}\amp,\sigma_K) } -1 \right) } \right\}, \quad k\in[K-1].\label{eq: Rk inequality2}
	\end{IEEEeqnarray}}Combining \eqref{eq: RK inequality2} and \eqref{eq: Rk inequality2}, we complete the proof.
\end{proof}

The high-SNR capacity region is derived based on the inner bound in Theorem~\ref{theorem: K-user inner bound only peak} and the outer bound in Theorem~\ref{theorem: K-user outer bound only peak}. We summarize it as follows.
\begin{theorem}[\textbf{Asymptotic Capacity Region}]\label{theorem: high SNR K-user capacity region only peak}
	When the input is only subject to the peak-intensity constraint in \eqref{eq:peak cons}, at high SNR, the capacity region of a $K$-user OI-BC asymptotically converges to the region where the rate tuple $(R_1,R_2,\cdots,R_K)$ satisfies
	{\setlength\abovedisplayskip{4.85pt} 
	\setlength\belowdisplayskip{4.85pt}
	\begin{IEEEeqnarray}{rCl}
		R_k & \ \dot{\leq}\ & \frac{1}{2} \log \left( \frac{ {\rho}_k^2\amp^2 +  2\pi e \sigma_k^2  }{ {\rho}_{k-1}^2\amp^2 +  2\pi e \sigma_k^2 } \right),\quad k\in[K], \label{eq: high SNR K-user only peak}
	\end{IEEEeqnarray}}with $\bm{\rho}\in \mathfrak{D}_{\bm{\rho}}$ and $\mathfrak{D}_{\bm{\rho}}=\left\{\bm{\rho}: {\rho}_k \in[0,1], \ {\rho}_{k-1}\leq {\rho}_{k}, \ \forall k\in[K], \ {\rho}_0=0,\ {\rho}_K=1\right\}$.
\end{theorem}
\begin{proof}
	Substituting the single-user capacity result in Lemma~\ref{lemma: 1-user capacity only peak} into Theorem~\ref{theorem: K-user outer bound only peak}, we obtain that the capacity region at high SNR is outer bounded by $\cup_{\bm{\rho}\in\mathfrak{D}_{\bm{\rho}} } \left( R_{1}^{\textnormal{out}}(\bm{\rho}),R_{2}^{\textnormal{out}}(\bm{\rho}),\cdots,R_{K}^{\textnormal{out}}(\bm{\rho}) \right)$, where
	{\setlength\abovedisplayskip{4.85pt} 
	\setlength\belowdisplayskip{4.85pt}
	\begin{IEEEeqnarray}{rCl}
		R_{k}^{\textnormal{out}}(\bm{\rho}) 
		& \doteq & \frac{1}{2} \log \left( \frac{ {\rho}_k^2\amp^2 +  2\pi e \sigma_k^2  }{ {\rho}_{k-1}^2\amp^2 +  2\pi e \sigma_k^2 } \right),\quad k\in[K].\label{eq: high-SNR K-uer Rk UB only peak}
	\end{IEEEeqnarray}}We write the boundary of the above outer bound as
	{\setlength\abovedisplayskip{4.85pt} 
	\setlength\belowdisplayskip{4.85pt}
\begin{IEEEeqnarray}{rCl}
		R_{K}^{\textnormal{out}}
		&\doteq& \frac{1}{2} \log\left( 1+ \frac{\amp^2}{2\pi e \sigma_K^2} \right) 
		- \frac{1}{2}\log \left(  1 + \sum_{m=1}^{K-1} \frac{\sigma_m^2}{\sigma_K^2}\left( e^{2 R_{m}^{\textnormal{out}} } -1 \right) \prod_{n=m+1}^{K-1} e^{ 2 R_{n}^{\textnormal{out}} } \right),\qquad \label{eq: high-SNR OB boundary only peak}
	\end{IEEEeqnarray}}whose proof is given in Appendix~\ref{app: proof of eq. 2}. The derivative of $R_{K}^{\textnormal{out}}$ is given by
	{\setlength\abovedisplayskip{4.85pt} 
	\setlength\belowdisplayskip{4.85pt}
	\begin{IEEEeqnarray}{rCl}
		\frac{\partial R_{K}^{\textnormal{out}} }{ \partial R_{k}^{\textnormal{out}} }
		&\doteq& - \frac{\sum_{m=1}^{k-1} \frac{\sigma_m^2}{\sigma_K^2}\left(e^{2 R_{m}^{\textnormal{out}}}-1\right) \prod_{n=m+1}^{K-1} e^{2 R_{n}^{\textnormal{out}}}
			+ \frac{\sigma_k^2}{\sigma_K^2}  \prod_{n=k}^{K-1} e^{2 R_{n}^{\textnormal{out}}}}{ 1+\sum_{m=1}^{K-1} \frac{\sigma_m^2}{\sigma_K^2}\left(e^{2 R_{m}^{\textnormal{out}}}-1\right) \prod_{n=m+1}^{K-1} e^{2 R_{n}^{\textnormal{out}}} }\\
		&=& - \frac{ \frac{\sigma_1^2}{\sigma_K^2}\left(e^{2 R_{1}^{\textnormal{out}}}-1\right) \prod_{n=2}^{K-1} e^{2 R_{n}^{\textnormal{out}}}
			+ \sum_{m=2}^{k-1} \frac{\sigma_m^2}{\sigma_K^2}\left(e^{2 R_{m}^{\textnormal{out}}}-1\right) \prod_{n=m+1}^{K-1} e^{2 R_{n}^{\textnormal{out}}}
			+ \frac{\sigma_k^2}{\sigma_K^2}  \prod_{n=k}^{K-1} e^{2 R_{n}^{\textnormal{out}}}}{ 1+\sum_{m=1}^{K-1} \frac{\sigma_m^2}{\sigma_K^2}\left(e^{2 R_{m}^{\textnormal{out}}}-1\right) \prod_{n=m+1}^{K-1} e^{2 R_{n}^{\textnormal{out}}} }\nonumber\\
		&\doteq& - \frac{ \frac{\sigma_1^2}{\sigma_K^2} \frac{{\rho}_1^2\amp^2}{2\pi e\sigma_1^2} \prod_{n=2}^{K-1} \frac{{\rho}_n^2}{{\rho}_{n-1}^2}
			+ \sum_{m=2}^{k-1} \frac{\sigma_m^2}{\sigma_K^2}\left( \frac{{\rho}_m^2}{{\rho}_{m-1}^2}-1\right) \prod_{n=m+1}^{K-1} \frac{{\rho}_n^2}{{\rho}_{n-1}^2}
			+ \frac{\sigma_k^2}{\sigma_K^2}  \prod_{n=k}^{K-1} \frac{{\rho}_n^2}{{\rho}_{n-1}^2} }{ 1 + \frac{{\rho}_{K-1}^2\amp^2}{2\pi e\sigma_K^2} } \label{eq: 121}\\
		&=& - \frac{\frac{{\rho}_{K-1}^2\amp^2}{2\pi e\sigma_K^2}
			+ \sum_{m=2}^{k-1} \frac{\sigma_m^2}{\sigma_K^2} \left( \frac{{\rho}_m^2}{{\rho}_{m-1}^2}-1\right) \frac{{\rho}_{K-1}^2}{{\rho}_{m}^2}
			+ \frac{\sigma_k^2}{\sigma_K^2}  \frac{{\rho}_{K-1}^2}{{\rho}_{k-1}^2} }{ 1 + \frac{{\rho}_{K-1}^2\amp^2}{2\pi e\sigma_K^2} } \\
		&\doteq& -1 , \quad k\in[K-1], \label{eq: slope}
	\end{IEEEeqnarray}}where \eqref{eq: 121} follows from \eqref{eq: high-SNR K-uer Rk UB only peak} and \eqref{eq: 135}. 
	Hence, at high SNR, the boundary of the outer bound is a $(K-1)$-dimensional hyperplane. For convenience, we denote it by $\mathcal{H}_\textnormal{out}$.

	Next, we proceed to show that the inner bound in Theorem~\ref{theorem: K-user inner bound only peak} is tight with $\mathcal{H}_\textnormal{out}$ at high SNR. Substituting \eqref{eq: only peak high-SNR capacity} into \eqref{eq: Rk Lb only peak K-user}, we obtain that
	 the rate tuples in $\mathsf{Conv}\{ \cup_{\bm{N}\in\mathfrak{D}_{\bm{N}}} \left( R_{1}^{\textnormal{in}}(\bm{N}),\cdots,R_{K}^{\textnormal{in}}(\bm{N}) \right) \}$ are all achievable, where
	{\setlength\abovedisplayskip{4.85pt} 
	\setlength\belowdisplayskip{4.85pt}
	\begin{IEEEeqnarray}{rCl}
	R_{k}^{\textnormal{in}}(\bm{N}) 
	&=& \frac{1}{2} \log \left( \frac{ \frac{ \amp^2 }{ (\prod_{n=k}^K {N}_{n})^2 } + 2 \pi e \sigma_k^2 }{ \frac{\amp^2}{(\prod_{n={k-1}}^K {N}_{n})^2} + 2 \pi e \sigma_k^2 } \right) , \quad k\in[K]. \label{eq: high-SNR K-uer Rk LB only peak}
	\end{IEEEeqnarray}}To analyze the achievable region $\mathsf{Conv}\{ \cup_{\bm{N}\in\mathfrak{D}_{\bm{N}}} \left( R_{1}^{\textnormal{in}}(\bm{N}),\cdots,R_{K}^{\textnormal{in}}(\bm{N}) \right) \}$, we define a new region $\mathcal{R}$ such that
	{\setlength\abovedisplayskip{4.85pt} 
	\setlength\belowdisplayskip{4.85pt}
	\begin{IEEEeqnarray}{rCl}
		\mathcal{R} =  \cup_{\bm{N}\in\mathfrak{D}_{\bm{N}}} \left( R_{1}^{\textnormal{in}}(\bm{N}),\cdots,R_{K}^{\textnormal{in}}(\bm{N}) \right). 
	\end{IEEEeqnarray}}Since $\mathcal{R}\subseteq \mathsf{Conv}\{ \cup_{\bm{N}\in\mathfrak{D}_{\bm{N}}} \left( R_{1}^{\textnormal{in}}(\bm{N}),\cdots,R_{K}^{\textnormal{in}}(\bm{N}) \right) \}$, thus the rate tuples in $\mathcal{R}$ are also achievable. Then, we have two observations about $\mathcal{R}$. First, given any $\bm{N}\in\mathfrak{D}_{\bm{N}}$, if $\bm{\rho}$ satisfies
	{\setlength\abovedisplayskip{4.85pt} 
	\setlength\belowdisplayskip{4.85pt}
	\begin{IEEEeqnarray}{rCl}
		{\rho}_k = \frac{1}{ \prod_{n=k}^K {N}_{n} }, \quad k \in \{0\cup[K]\}, \label{eq: rho and N relationship}
	\end{IEEEeqnarray}}then combined with \eqref{eq: high-SNR K-uer Rk UB only peak} and \eqref{eq: high-SNR K-uer Rk LB only peak}, we have
	{\setlength\abovedisplayskip{4.85pt} 
	\setlength\belowdisplayskip{4.85pt}
	\begin{IEEEeqnarray}{rCl}
		\left( R_{1}^{\textnormal{in}}(\bm{N}),\cdots,R_{K}^{\textnormal{in}}(\bm{N}) \right) \doteq \left( R_{1}^{\textnormal{out}}( \bm{\rho}),\cdots, R_{K}^{\textnormal{out}}( \bm{\rho}) \right), \quad \forall \bm{N} \in\mathfrak{D}_{\bm{N}}. \label{eq: equal rates}
	\end{IEEEeqnarray}}Therefore, any achievable rate tuple in $\mathcal{R}$, is equal to some rate tuple on the hyperplane $\mathcal{H}_\textnormal{out}$ at high SNR, i.e., $\mathcal{R}\subseteq\mathcal{H}_\textnormal{out}$. Second, given any $k\in[K]$, if we let
	{\setlength\abovedisplayskip{4.85pt} 
	\setlength\belowdisplayskip{4.85pt}
	\begin{IEEEeqnarray}{rCl}
		{N}_{k-1} &\rightarrow & +\infty,\\
		{N}_{n} &\rightarrow & 1, \quad n\in\{k,\cdots,K\},
	\end{IEEEeqnarray}}then combined with \eqref{eq: high-SNR K-uer Rk LB only peak}, we can obtain that the following tuples in $\mathcal{R}$ are achievable:
	{\setlength\abovedisplayskip{4.85pt} 
	\setlength\belowdisplayskip{4.85pt}
	\begin{IEEEeqnarray}{rCl}
		\Bigl(0,\cdots,0,\underbrace{\frac{1}{2} \log\left(1 + \frac{\amp^2}{2\pi e\sigma_k^2}\right)}_k,0,\cdots,0\Bigr),\quad k\in[K]. \label{eq1: boundary points only peak}
	\end{IEEEeqnarray}}From \eqref{eq: high-SNR K-uer Rk UB only peak}, we find that these achievable rate tuples in \eqref{eq1: boundary points only peak} are exactly the corner points on the hyperplane $\mathcal{H}_\textnormal{out}$.
	
	Now, we are ready to analyze the achievable region $\mathsf{Conv}\{ \cup_{\bm{N}\in\mathfrak{D}_{\bm{N}}} \left( R_{1}^{\textnormal{in}}(\bm{N}),\cdots,R_{K}^{\textnormal{in}}(\bm{N}) \right) \}$. Note that 
	{\setlength\abovedisplayskip{4.85pt} 
	\setlength\belowdisplayskip{4.85pt}
	\begin{IEEEeqnarray}{rCl}
		\mathsf{Conv} \left\{ \cup_{\bm{N}\in\mathfrak{D}_{\bm{N}}} \left( R_{1}^{\textnormal{in}}(\bm{N}),\cdots,R_{K}^{\textnormal{in}}(\bm{N}) \right) \right\} = \mathsf{Conv} \left\{ \mathcal{R} \right\}. \label{eq: 124}
	\end{IEEEeqnarray}}With the above two observations about $\mathcal{R}$, i.e., (1) $\mathcal{R}\subseteq\mathcal{H}_\textnormal{out}$; (2) the corner points on the hyperplane $\mathcal{H}_\textnormal{out}$ can be achieved by some rate tuples in $\mathcal{R}$, we conclude that $\mathsf{Conv} \left\{ \cup_{\bm{N}\in\mathfrak{D}_{\bm{N}}} \left( R_{1}^{\textnormal{in}}(\bm{N}),\cdots,R_{K}^{\textnormal{in}}(\bm{N}) \right) \right\}$ is also a $(K-1)$-dimensional hyperplane and overlaps with $\mathcal{H}_\textnormal{out}$. Hence, the inner bound in Theorem~\ref{theorem: K-user inner bound only peak} and the outer bound in Theorem~\ref{theorem: K-user outer bound only peak} are tight at high SNR, which completes the proof.	

\end{proof}


The derived bounds on the capacity region are shown in Fig.~\ref{fig: K-user only peak}, where we assume $K=3$, $\sigma_3 = 2\sigma_2=4\sigma_1$, and $\mathrm{ASNR}_k=\frac{\amp}{\sigma_k}$, $k\in\{1,2,3\}$. In the figure, we first depict the $3$-dimensional capacity region bounds, then project the result onto $R_2$ and $R_3$ plane, $R_2$ and $R_3$ plane, and $R_2$ and $R_3$ plane, respectively. From Fig.~\ref{fig: K-user only peak}, we observe that the inner and outer bounds become tighter as SNR increases, which further validates the derivations.

\begin{figure}[!htbp]
	\centering
	\subfigure[Bounds on capacity region.]{
		\includegraphics[width=3.3in]{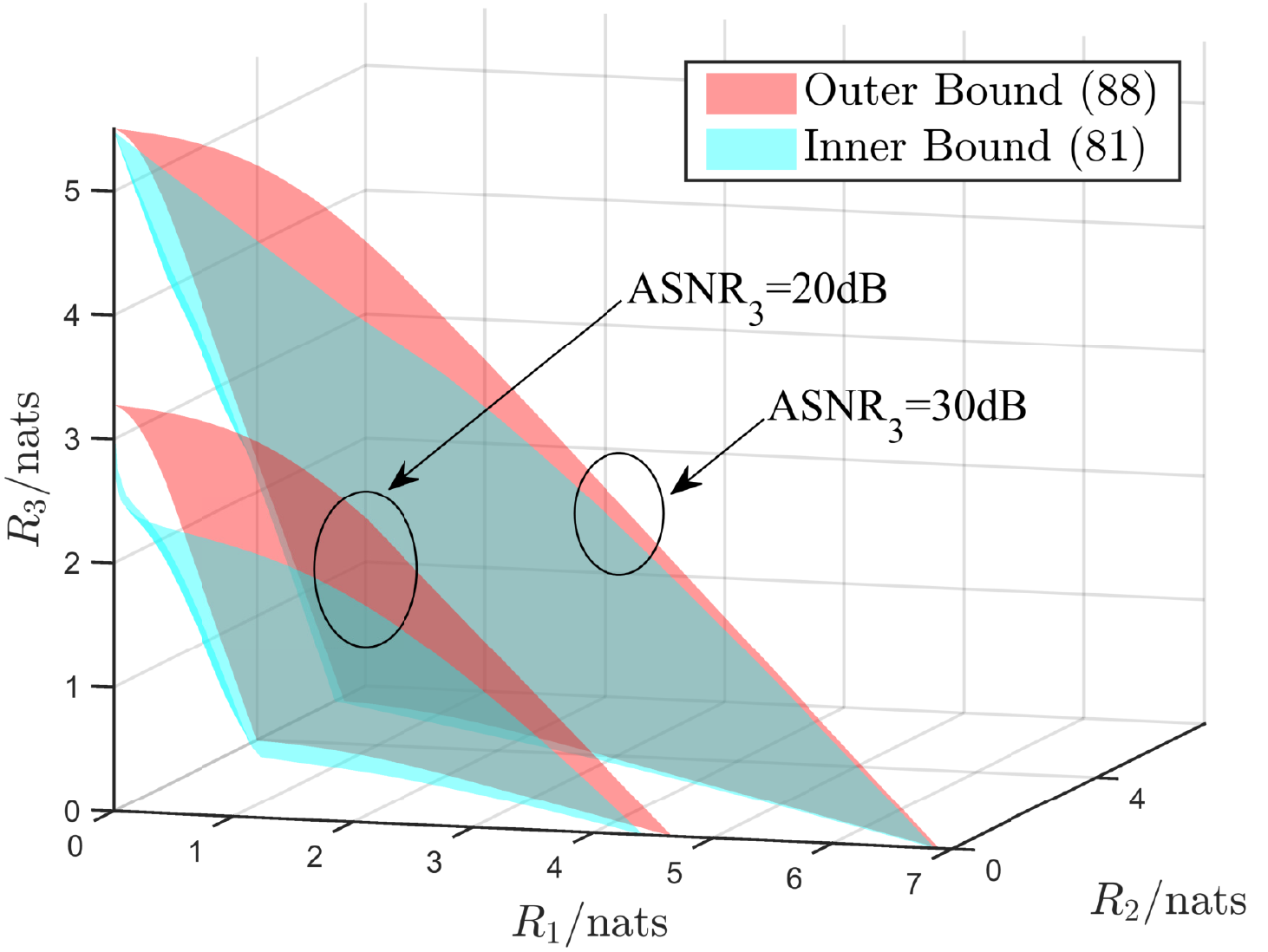}	
	}\hspace{-2.5mm}
	\subfigure[Projection on the $R_1$ and $R_2$ plane.]{
		\includegraphics[width=3.2in]{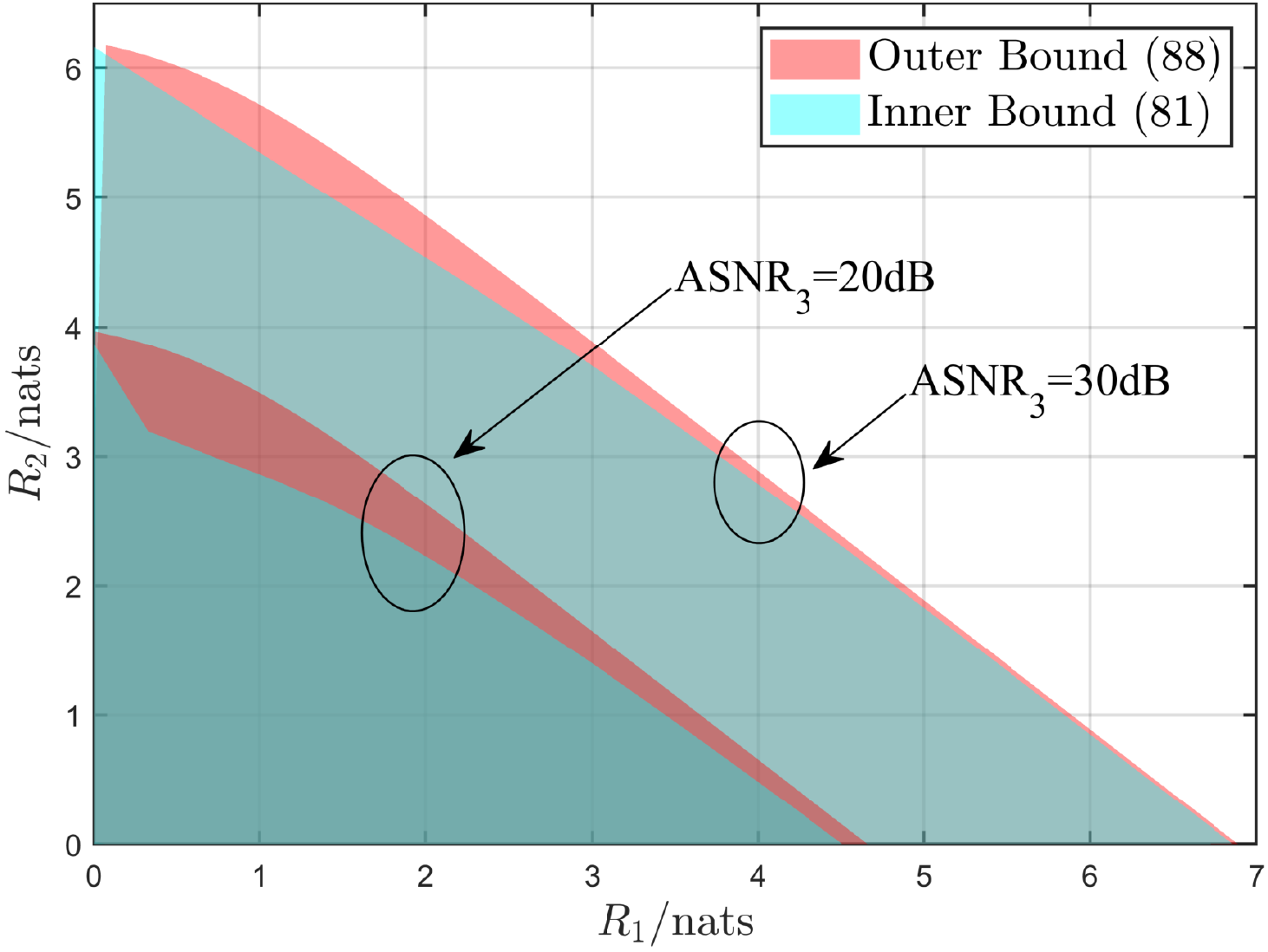}
	}
	\subfigure[Projection on the $R_2$ and $R_3$ plane.]{
		\includegraphics[width=3.2in]{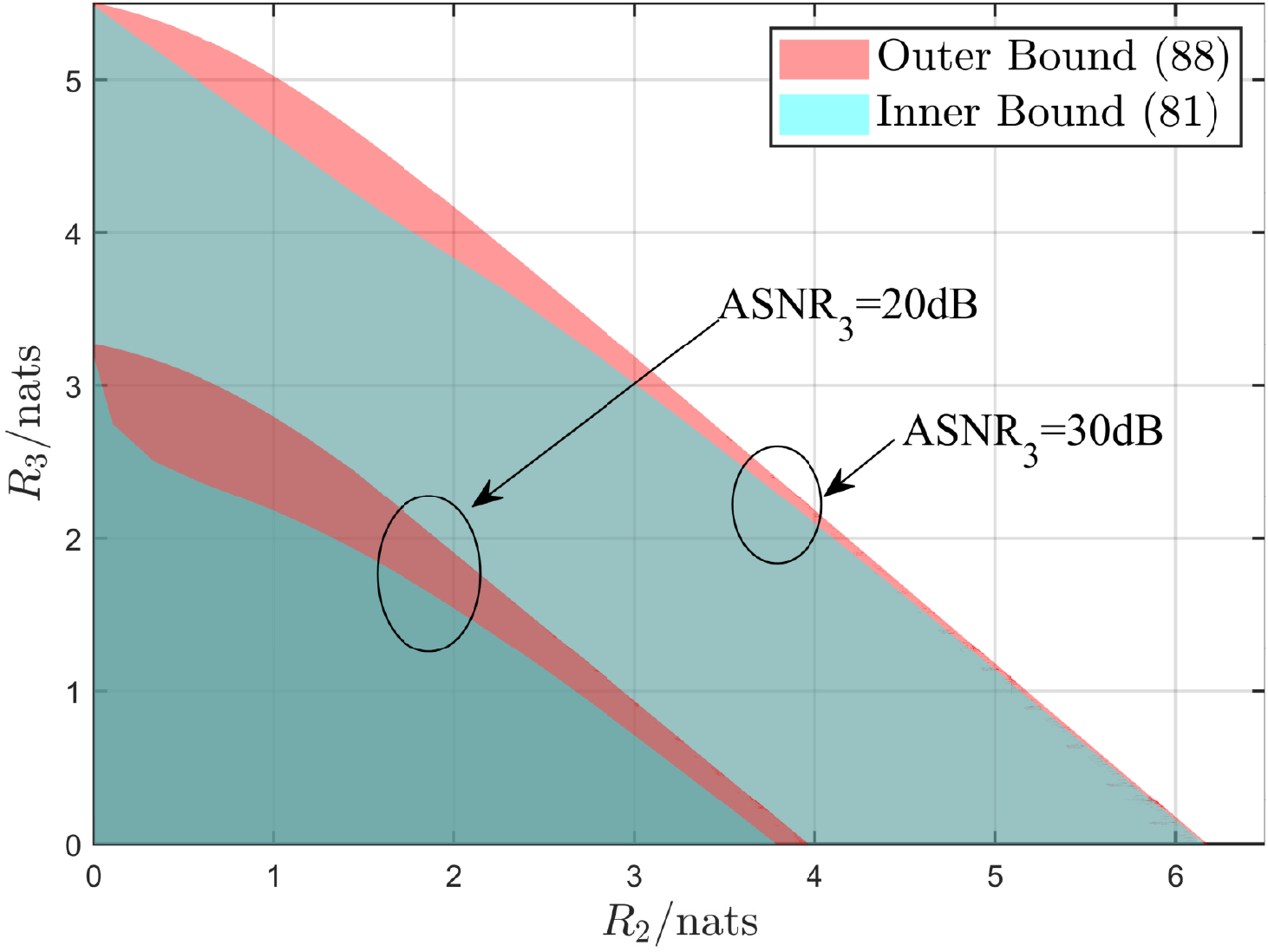}
	}\hspace{-0mm}
	\subfigure[Projection on the $R_1$ and $R_3$ plane.]{
		\includegraphics[width=3.2in]{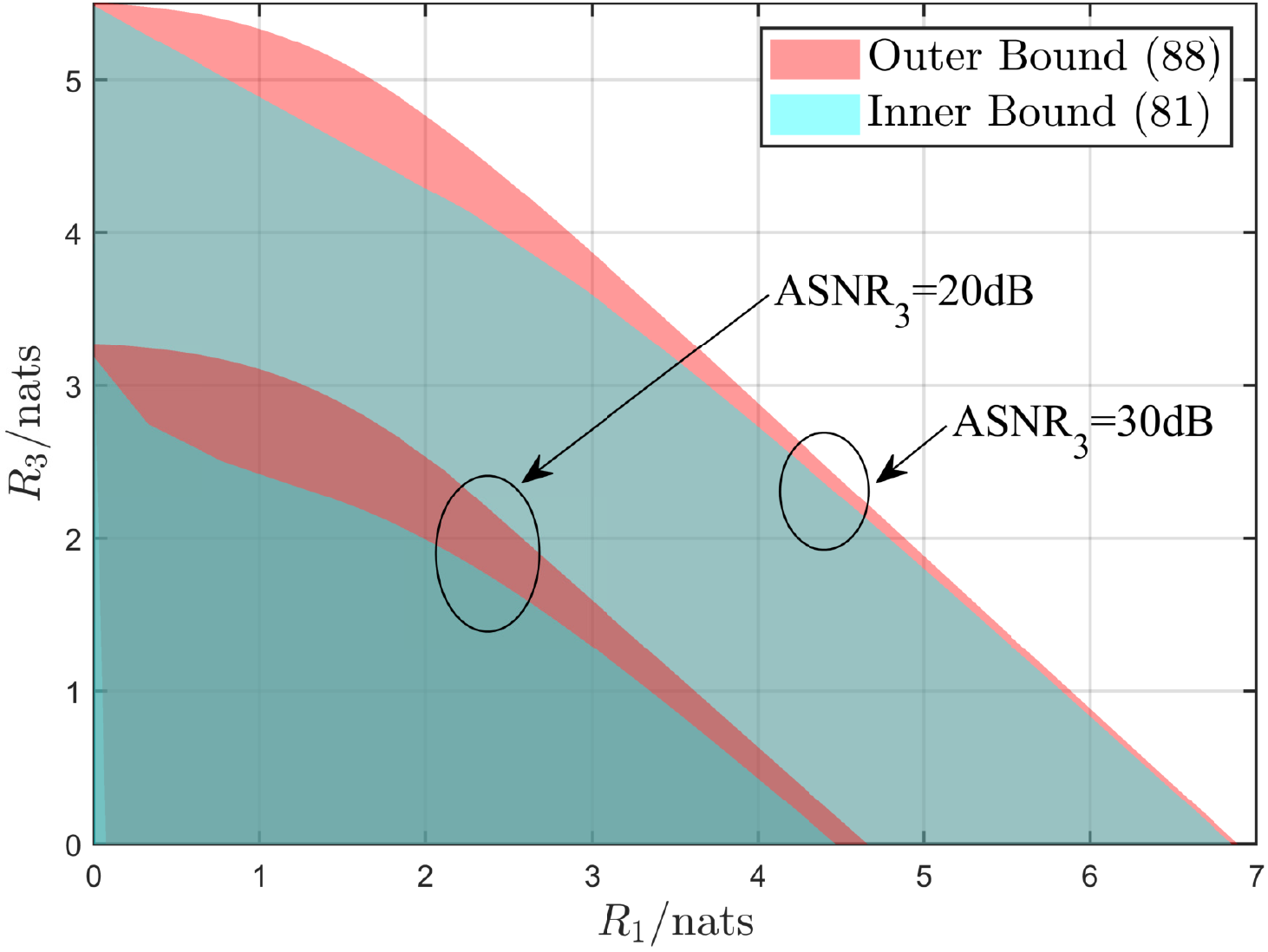}
	}	
	\caption{Bounds on capacity region of $3$-user OI-BC with peak-intensity constraint.}
	\label{fig: K-user only peak}	
	\vspace{-0.7cm}
\end{figure}

\subsection{Average-Intensity Constrained OI-BC}
The inner bound on the capacity region is proposed as follows.
\begin{theorem}[\textbf{Inner Bound}]\label{theorem: K-user inner bound only ave}
	When the input is only subject to the average-intensity constraint in \eqref{eq:ave cons}, the rate tuples in the set $\cup_{\bm{\rho}\in\mathfrak{D}_{\bm{\rho}}} \left( R_{1}^{\textnormal{in}}(\bm{\rho}),R_{2}^{\textnormal{in}}(\bm{\rho}),\cdots,R_{K}^{\textnormal{in}}(\bm{\rho}) \right)$ are all achievable for a $K$-user OI-BC, where
	{\setlength\abovedisplayskip{4.85pt} 
	\setlength\belowdisplayskip{4.85pt}
	\begin{IEEEeqnarray}{rCl}
		R_{k}^{\textnormal{in}}(\bm{\rho}) &=& \frac{1}{2} \log \left( 1+\frac{e{\rho}_k^2\EE^2}{2 \pi \sigma_k^2} \right) -\mathsf{C}_{\textnormal{ub}}\left( {\rho}_{k-1}\EE,\sigma_k\right), \quad k\in[K], \qquad\label{eq: Rk Lb only average K-user}
	\end{IEEEeqnarray}}with $\mathfrak{D}_{\bm{\rho}}=\left\{\bm{\rho}=[{\rho}_0,{\rho}_1,\cdots,{\rho}_K]: {\rho}_{k}\in[0,1],\ {\rho}_{k-1}\leq {\rho}_{k}, \ \forall k\in[K], \ {\rho}_0=0,\ {\rho}_K=1\right\}$.
\end{theorem} 
\begin{proof}
	Assume $X_1$ follows $\textsf{Exp}(\rho_1\EE)$ and $\forall  k\in\{2,3,\cdots,K\}$,
	{\setlength\abovedisplayskip{4.85pt} 
	\setlength\belowdisplayskip{4.85pt}
	\begin{IEEEeqnarray}{rCl}
		p_{X_k}(x_k) 
		&=& \frac{{\rho}_{k-1}}{{\rho}_k} \delta(x_k) + \left( 1-\frac{{\rho}_{k-1}}{{\rho}_k} \right) \times \frac{1}{{\rho}_k\EE}e^{-\frac{x_k}{{\rho}_k\EE}},\ x_k\in[0,+\infty).
	\end{IEEEeqnarray}}With Lemma~\ref{lemma: decom of uniform}, we obtain that $X_{k,\textnormal{sum}}$ in \eqref{eq: X_sum def} follows $\textsf{Exp}(\rho_k\EE)$ and $X$ follows $\textsf{Exp}(\EE)$, which satisfies the average-intensity constraint in \eqref{eq:ave cons}. Then, following similar steps in the proof of Theorem~\ref{theorem: K-user inner bound only peak}, we can complete the proof.
\end{proof}

We propose the outer bound on the capacity region in Theorem~\ref{theorem: K-user outer bound only ave}, whose proof follows similar arguments as in the proof of Theorem~\ref{theorem: K-user outer bound only peak} and omitted here.	
\begin{theorem}[\textbf{Outer Bound}]\label{theorem: K-user outer bound only ave}
	When the input is only subject to the average-intensity constraint in \eqref{eq:ave cons}, the capacity region of a $K$-user OI-BC is outer bounded by\\
	$\cup_{\bm{\rho}\in\mathfrak{D}_{\bm{\rho}}} \left( R_{1}^{\textnormal{out}}(\bm{\rho}),R_{2}^{\textnormal{out}}(\bm{\rho}),\cdots,R_{K}^{\textnormal{out}}(\bm{\rho}) \right)$, where
	{\setlength\abovedisplayskip{4.85pt} 
	\setlength\belowdisplayskip{4.85pt}
	\begin{IEEEeqnarray}{rCl}
		R_{k}^{\textnormal{out}}(\bm{\rho}) &=& \frac{1}{2} \log \left\{\frac{ \sigma_k^2 + \sigma_K^2 \left( e^{ 2\mathsf{C}_{\textnormal{ub}}({\rho}_k\EE,\sigma_K) } -1 \right) }{ \sigma_k^2 + \sigma_K^2 \left( e^{ 2\mathsf{C}_{\textnormal{ub}}({\rho}_{k-1}\EE,\sigma_K) } -1 \right) } \right\},\quad k\in[K],\label{eq: Ub only ave K-user}
	\end{IEEEeqnarray}}with $\mathfrak{D}_{\bm{\rho}}=\left\{\bm{\rho}=[{\rho}_0,{\rho}_1,\cdots,{\rho}_K]: {\rho}_{k}\in[0,1],\ {\rho}_{k-1}\leq {\rho}_{k}, \ \forall k\in[K], \ {\rho}_0=0,\ {\rho}_K=1\right\}$.
\end{theorem}

The high-SNR capacity region is directly obtained by substituting the single-user capacity result in Lemma~\ref{lemma: 1-user capacity only ave} into \eqref{eq: Rk Lb only average K-user} and \eqref{eq: Ub only ave K-user}. We summarize it as follows.

\begin{theorem}[\textbf{Asymptotic Capacity Region}]\label{theorem: high SNR K-user capacity region only ave}
	When the input is only subject to the average-intensity constraint in \eqref{eq:ave cons}, at high SNR, the capacity region of a $K$-user OI-BC asymptotically converges to the region where the rate tuple $\left( {R}_1,{R}_2,\cdots,{R}_K \right)$ satisfies
	{\setlength\abovedisplayskip{4.85pt} 
	\setlength\belowdisplayskip{4.85pt}
	\begin{IEEEeqnarray}{rCl}
		R_k & \ \dot{\leq}\ & \frac{1}{2} \log \left( \frac{ e{\rho}_k^2\amp^2 +  2\pi \sigma_k^2  }{ e{\rho}_{k-1}^2\amp^2 +  2\pi \sigma_k^2 } \right),\quad k\in[K],
	\end{IEEEeqnarray}}with $\bm{\rho}\in \mathfrak{D}_{\bm{\rho}}$ and $\mathfrak{D}_{\bm{\rho}}=\left\{\bm{\rho}: {\rho}_k \in[0,1], \ {\rho}_{k-1}\leq {\rho}_{k}, \ \forall k\in[K], \ {\rho}_0=0,\ {\rho}_K=1\right\}$.
\end{theorem} 

The derived bounds on the capacity region are shown in Fig.~\ref{fig: K-user only ave}, where we assume $K=3$, $\sigma_3 = 2\sigma_2=4\sigma_1$, and $\mathrm{ESNR}_k=\frac{\EE}{\sigma_k}$, $k\in\{1,2,3\}$. 

\begin{figure}[!htbp]
	\centering
	\subfigure[Bounds on capacity region.]{
		\includegraphics[width=3.3in]{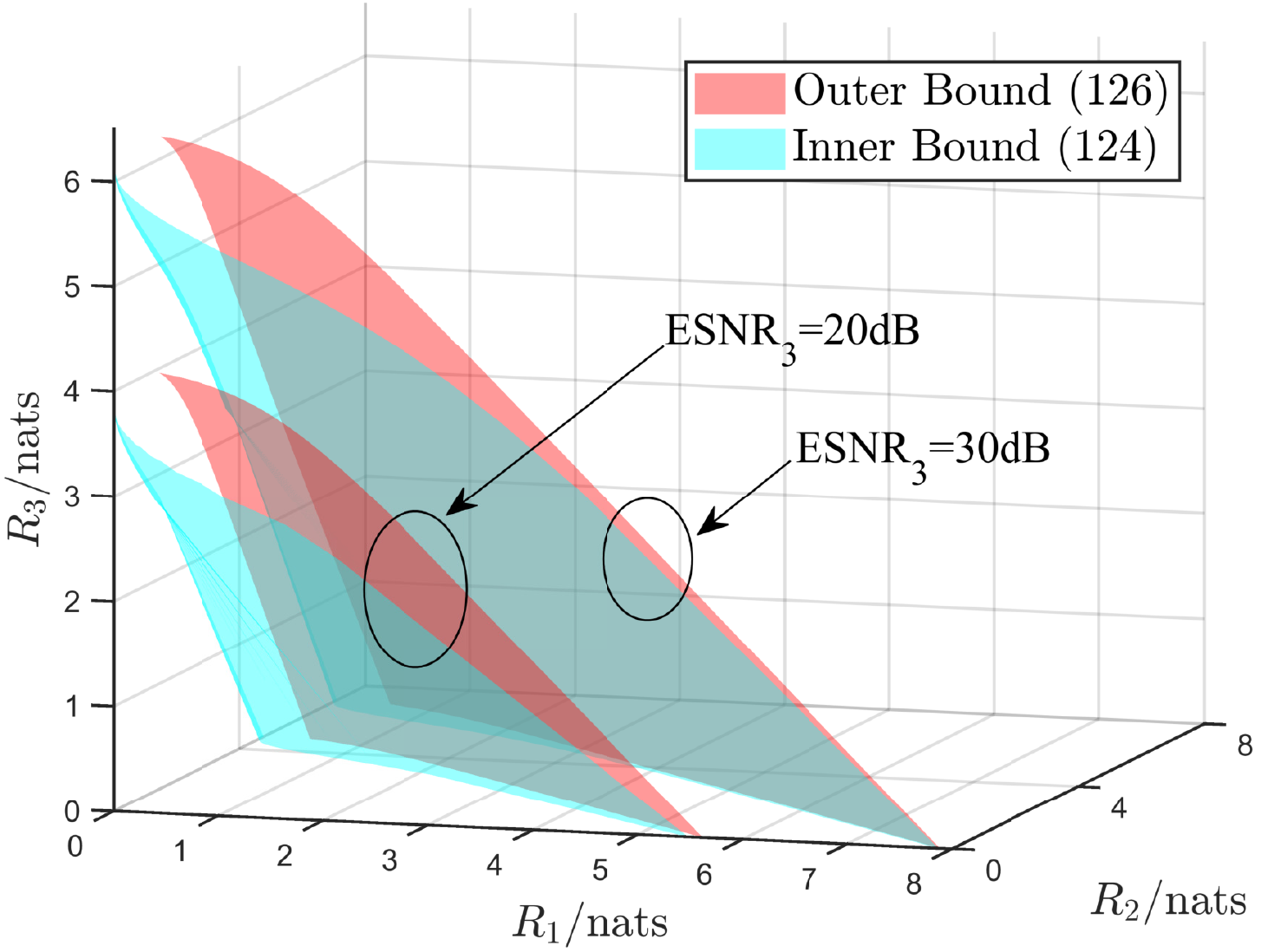}
	}\hspace{-2.5mm}
	\subfigure[Projection on the $R_1$ and $R_2$ plane.]{
		\includegraphics[width=3.2in]{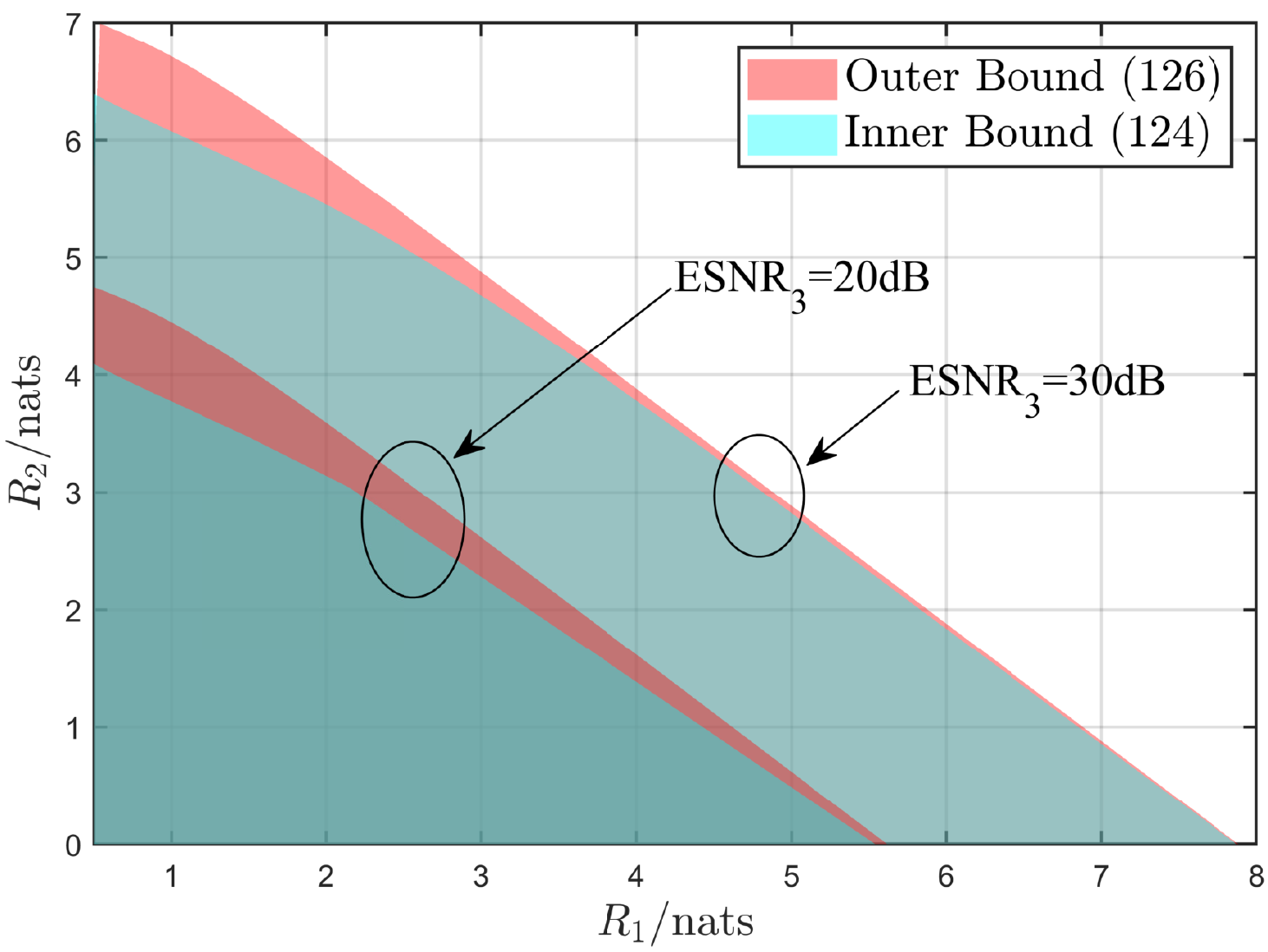}
	}
	\subfigure[Projection on the $R_2$ and $R_3$ plane.]{
		\includegraphics[width=3.2in]{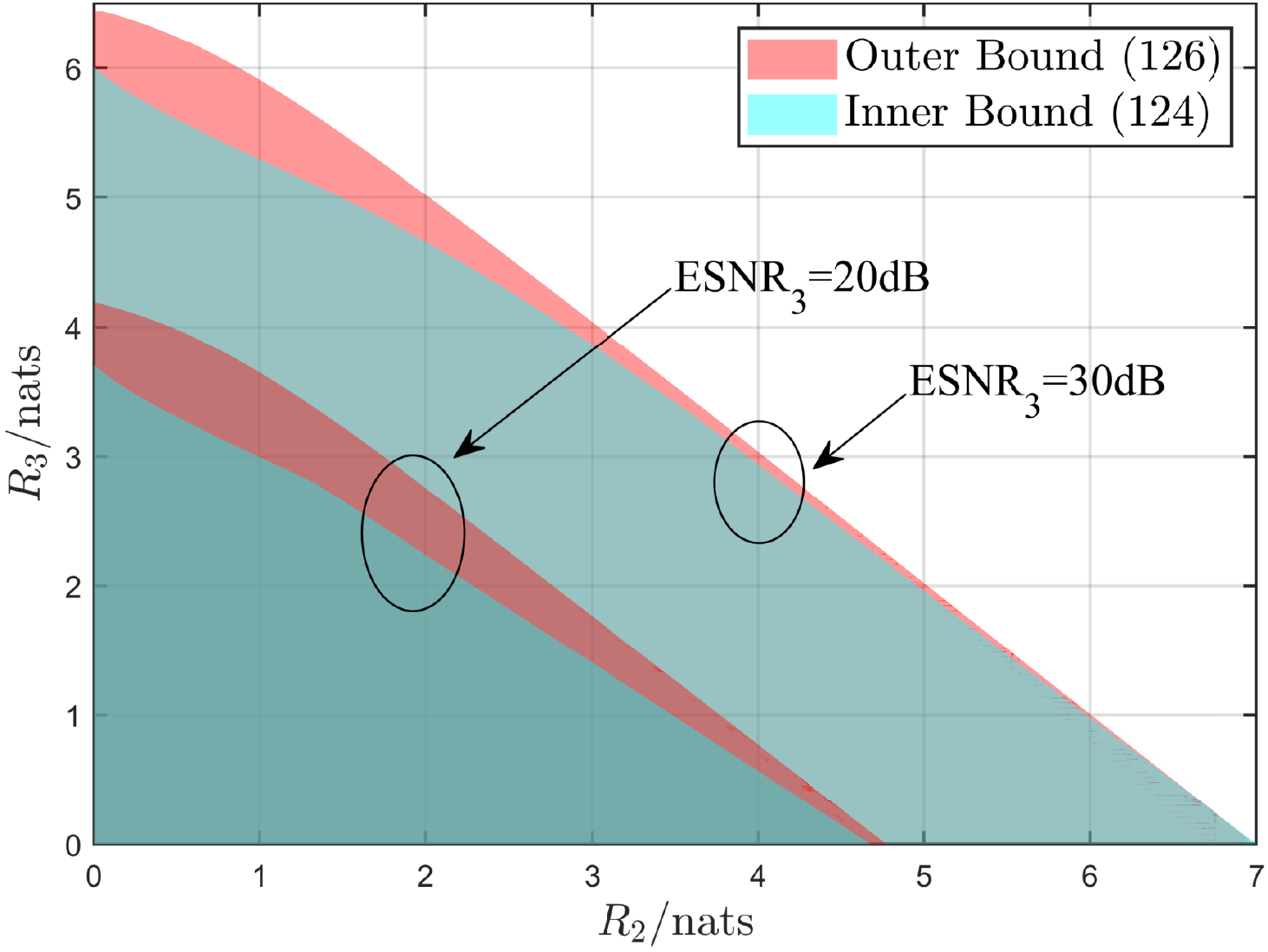}
	}\hspace{-0mm}
	\subfigure[Projection on the $R_1$ and $R_3$ plane.]{
		\includegraphics[width=3.2in]{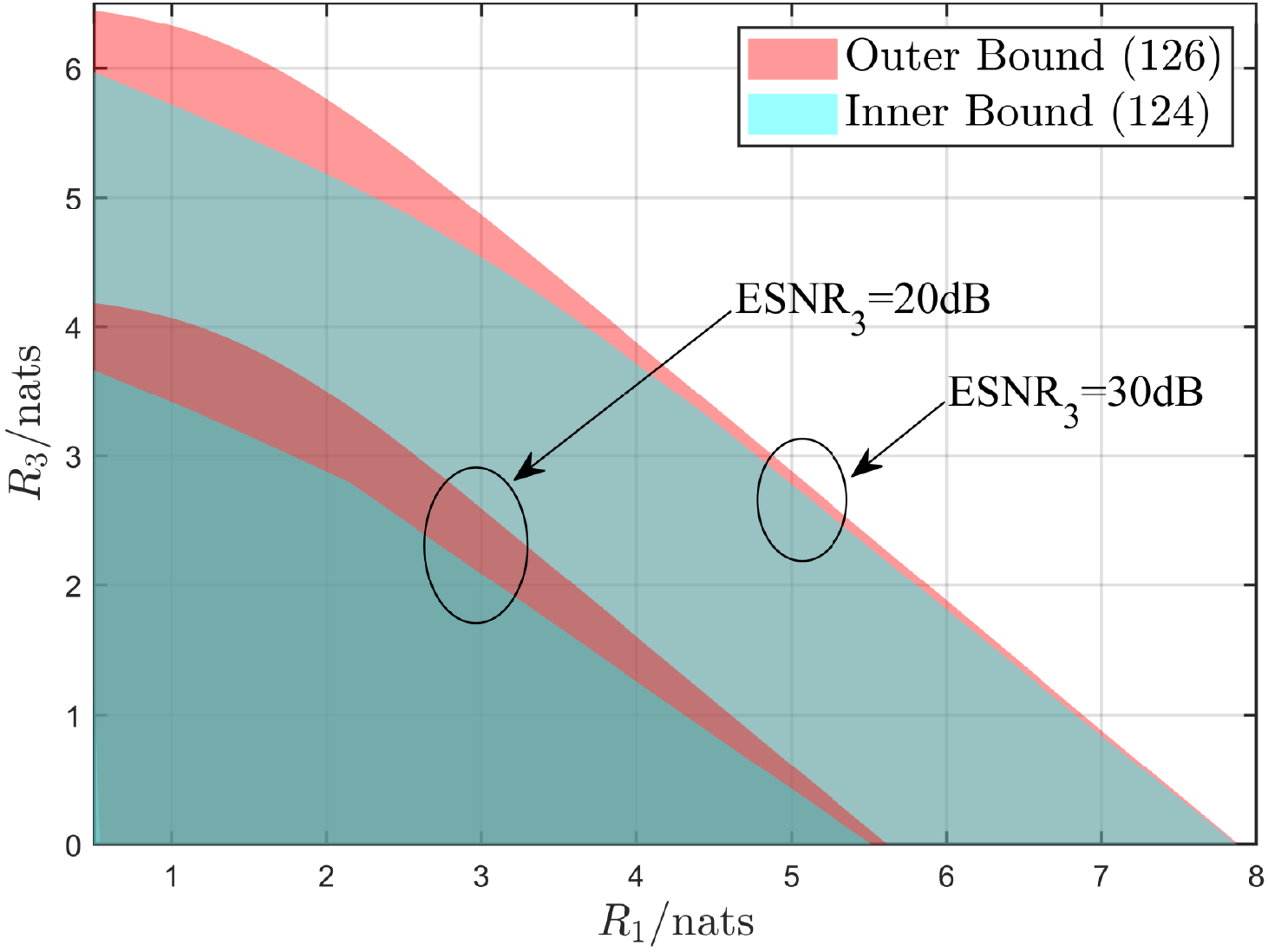}
	}	
	\caption{Bounds on capacity region of $3$-user OI-BC with average-intensity constraint.}
	\label{fig: K-user only ave}
	\vspace{-0.7cm}
\end{figure}

\subsection{Peak- and Average-Intensity Constrained OI-BC}
The inner bound on the capacity region is proposed as follows.
\begin{theorem}[\textbf{Inner Bound}]\label{theorem: K-user inner bound both power}
	When the input is subject to both peak- and average-intensity constraints in \eqref{eq:peak cons} and \eqref{eq:ave cons}, the rate tuples in the set $\mathsf{Conv}\{ \cup_{\bm{N}\in\mathfrak{D}_{\bm{N}}} \left( R_{1}^{\textnormal{in}}(\bm{N}),R_{2}^{\textnormal{in}}(\bm{N}),\cdots,R_{K}^{\textnormal{in}}(\bm{N}) \right) \}$ are all achievable for a $K$-user OI-BC, where
	{\setlength\abovedisplayskip{4.85pt} 
	\setlength\belowdisplayskip{4.85pt}
	\begin{IEEEeqnarray}{rCl}
		R_{k}^{\textnormal{in}}(\bm{N}) 
		&=& \frac{1}{2} \log \left\{ 1 + 
		\exp \left(2 - \frac{ \frac{2\mu^\star}{\prod_{n=k}^K {N}_{n}} e^{- \frac{\mu^\star}{\prod_{n=k}^K {N}_{n}} } }{ 1-e^{-\frac{\mu^\star}{\prod_{n=k}^K {N}_{n}} } }\right) \left( \frac{1- e^{- \frac{\mu^\star}{\prod_{n=k}^K {N}_{n}} } }{ \mu^\star }	\right)^2 
		\frac{\amp^2}{2 \pi e \sigma_k^2} \right\}\nonumber\\
		&& -\mathsf{C}_{\textnormal{ub}} \left( \frac{\amp}{\prod_{n=k-1}^K {N}_{n}},\sigma_k;\frac{\mu^\star}{\prod_{n=k-1}^K {N}_{n}} \right), \quad k\in[K], \qquad
		\label{eq: Rk Lb both power K-user}
	\end{IEEEeqnarray}}with $\mathfrak{D}_{\bm{N}}=\{ \bm{N}=[N_0,N_1,\cdots,N_K]: {N}_{n}\in \mathbb{N}^{+}, \ \forall n\in \{0\cup[K]\}, \ {N}_0\gg 0,\ {N}_K=1 \}$.
\end{theorem} 
\begin{proof}
	Assume $X_1$ follows $\textsf{Texp}\left(\frac{\amp}{\prod_{n=1}^K N_n},\frac{\mu^\star}{\prod_{n=1}^K N_n}\right)$ and $\forall k\in\{2,3,\cdots,K\}$,
	{\setlength\abovedisplayskip{4.85pt} 
	\setlength\belowdisplayskip{4.85pt}
	\begin{IEEEeqnarray}{rCl}
		p_{X_k}(x_k) 
		&=& \frac{ 1-e^{ - \frac{\frac{\mu^\star}{ \prod_{n=k}^KN_n }}{N_{k-1}} } }{ 1-e^{-\frac{\mu^\star}{ \prod_{n=k}^KN_n }} } 
		\sum_{n=0}^{N_{k-1}-1} 
		e^{- \frac{\frac{\mu^\star}{ \prod_{n=k}^KN_n } n }{N_{k-1}} } 
		\delta \left(x_k-\frac{n}{N_{k-1} }\times \frac{\amp}{\prod_{n=k}^KN_n} \right).\quad
	\end{IEEEeqnarray}}With Lemma~\ref{lemma: decom of uniform}, we obtain that $X_{k,\textnormal{sum}}$ in \eqref{eq: X_sum def} follows $\textsf{Texp} \left( \frac{\amp}{\prod_{n=k}^K N_n}, \frac{\mu^\star}{\prod_{n=k}^K N_n} \right)$ and $X$ follows $\textsf{Texp} ( \amp, \mu^\star )$, which satisfies the peak- and average-intensity constraints in \eqref{eq:peak cons} and \eqref{eq:ave cons}. Then, similar to the proof of Theorem~\ref{theorem: K-user inner bound only peak}, we can complete the proof.
\end{proof}

The following outer bound has been given in \cite[Theorem 5]{Chaaban2016-2}. We can also utilize the conditional EPI to provide a new proof, which is similar to the proof of Theorem~\ref{theorem: K-user outer bound only peak} and omitted here.
\begin{lemma}[\textbf{Outer Bound}]\label{lemma: K-user outer bound both peak and ave}
	When the input is subject to both peak- and average-intensity constraints in \eqref{eq:peak cons} and \eqref{eq:ave cons}, the capacity region of a $K$-user OI-BC is outer bounded by $\cup_{\bm{\rho}\in\mathfrak{D}_{\bm{\rho}}} \left( R_{1}^{\textnormal{out}}(\bm{\rho}),R_{2}^{\textnormal{out}}(\bm{\rho})\right.,\cdots,$ $\left.R_{K}^{\textnormal{out}}(\bm{\rho}) \right)$, where
	{\setlength\abovedisplayskip{4.85pt} 
	\setlength\belowdisplayskip{4.85pt}
	\begin{IEEEeqnarray}{rCl}
		R_{k}^{\textnormal{out}}(\bm{\rho}) &=& \frac{1}{2} \log \left\{\frac{ \sigma_k^2 + \sigma_K^2 \left( e^{ 2\mathsf{C}_{\textnormal{ub}}({\rho}_k\amp,\sigma_K;\mu^\star ) } -1 \right) }{ \sigma_k^2 + \sigma_K^2 \left( e^{ 2\mathsf{C}_{\textnormal{ub}} ({\rho}_{k-1}\amp,\sigma_K;\mu^\star) } -1 \right) } \right\},\quad k\in[K],\label{eq: Ub both power K-user}
	\end{IEEEeqnarray}}with $\mathfrak{D}_{\bm{\rho}}=\left\{\bm{\rho}=[{\rho}_0,{\rho}_1,\cdots,{\rho}_K]: {\rho}_{k}\in[0,1],\ {\rho}_{k-1}\leq {\rho}_{k}, \ \forall k\in[K], \ {\rho}_0=0,\ {\rho}_K=1\right\}$.
\end{lemma}

The high-SNR capacity region is proposed as follows.
\begin{theorem}[\textbf{Asymptotic Capacity Region}]\label{theorem: high SNR K-user capacity region both power}
	When the input is subject to both peak- and average-intensity constraints in  \eqref{eq:peak cons} and \eqref{eq:ave cons}, at high SNR, the capacity region of a $K$-user OI-BC asymptotically converges to the region where the rate tuple $\left( {R}_1,{R}_2,\cdots,{R}_K \right)$ satisfies
	{\setlength\abovedisplayskip{4.85pt} 
	\setlength\belowdisplayskip{4.85pt}
	\begin{IEEEeqnarray}{rCl}
		R_k & \ \dot{\leq}\ & \frac{1}{2}\log \left(\frac{ \exp  \left(2 - \frac{2\rho_k\mu^\star e^{-\rho_k\mu^\star}}{1-e^{-\rho_k\mu^\star}}\right) \left( \frac{1- e^{-\rho_k\mu^\star} }{\mu^\star}
			\right)^2 \amp^2 + 2\pi e\sigma_k^2 }{ \exp \left(2 - \frac{2\rho_{k-1}\mu^\star e^{-\rho_{k-1}\mu^\star}}{1-e^{-\rho_{k-1}\mu^\star}}\right) \left( \frac{1- e^{-\rho_{k-1}\mu^\star} }{\mu^\star}
			\right)^2 \amp^2 + 2\pi e\sigma_k^2 } \right), \label{eq: high SNR K-user both power}
	\end{IEEEeqnarray}}with $\bm{\rho}\in \mathfrak{D}_{\bm{\rho}}$ and $\mathfrak{D}_{\bm{\rho}}=\left\{\bm{\rho}: {\rho}_k \in[0,1], \ {\rho}_{k-1}\leq {\rho}_{k}, \ \forall k\in[K], \ {\rho}_0=0,\ {\rho}_K=1\right\}$.
\end{theorem} 
\begin{proof}
	The proof follows similar arguments as in the proof of Theorem~\ref{theorem: high SNR K-user capacity region only peak}. Before that, we need to derive a new outer bound which is valid in the high SNR regime. Note that
	{\setlength\abovedisplayskip{4.85pt} 
	\setlength\belowdisplayskip{4.85pt}
	\begin{IEEEeqnarray}{rCl}
	\mathsf{h}(Y_K|U_K)	&\geq& \mathsf{h}(Z_K),\\
	\mathsf{h}(Y_K|U_K) &\leq& \mathsf{C}_{\textnormal{ub}}(\amp,\sigma_K;\mu^\star) + \mathsf{h}(Z_K).
	\end{IEEEeqnarray}}At high SNR, $\mathsf{C}_{\textnormal{ub}}(\rho\amp,\sigma_K;\rho\mu^\star)$, $\rho\in[0,1]$, is monotonically increasing with respect to $\rho$ and approaches zeros when $\rho$ tends to $0$. Hence, there exists $\rho_{K-1}\in[0,1]$ such that
{\setlength\abovedisplayskip{4.85pt} 
\setlength\belowdisplayskip{4.85pt}
\begin{IEEEeqnarray}{rCl}
	\mathsf{h}(Y_K|U_K) &=& \mathsf{C}_{\textnormal{ub}}({\rho}_{K-1}\amp,\sigma_K;\rho_{K-1}\mu^\star) + \mathsf{h}(Z_K)\\
	&=& \frac{1}{2} \log \left( 2\pi e\sigma_K^2 + 2\pi e\sigma_K^2 ( e^{ 2 \mathsf{C}_{\textnormal{ub}}({\rho}_{K-1}\amp,\sigma_K;\rho_{K-1}\mu^\star) } -1 ) \right). \label{eq: 182}
\end{IEEEeqnarray}}Furthermore, similar to the steps from \eqref{eq: h(Yk|Uk) equality} to \eqref{eq: EPI3}, we can obtain that
{\setlength\abovedisplayskip{4.85pt} 
\setlength\belowdisplayskip{4.85pt}
\begin{IEEEeqnarray}{rCl}
	\mathsf{h}(Y_k|U_k) 
	&=& \frac{1}{2} \log \left( 2\pi e\sigma_k^2 + 2\pi e\sigma_K^2 ( e^{ 2 \mathsf{C}_{\textnormal{ub}}({\rho}_{k-1}\amp,\sigma_K;\rho_{k-1}\mu^\star) } -1 ) \right), \ k\in[K-1],\\
	\mathsf{h}(Y_k|U_{k+1}) 
	&\leq& \frac{1}{2} \log \left( 2\pi e\sigma_{k}^2 + 2\pi e\sigma_K^2 ( e^{ 2 \mathsf{C}_{\textnormal{ub}}( {\rho}_{k}\amp,\sigma_K;{\rho}_{k}\mu^\star ) } -1 ) \right), \ k\in[K-1],
\end{IEEEeqnarray}}where $\bm{\rho}=[\rho_1,\cdots,\rho_K]\in\mathfrak{D}_{\bm{\rho}}$.
Then, by Lemma~\ref{lemma: Kuser}, at high SNR, we have
{\setlength\abovedisplayskip{4.85pt} 
\setlength\belowdisplayskip{4.85pt}
\begin{IEEEeqnarray}{rCl}
	R_k &\leq& \mathsf{h}(Y_k|U_{k+1}) - \mathsf{h}(Y_k|U_{k})\\
	&\ \dot{\leq} \ & \frac{1}{2}\log \left(\frac{ \exp \left( 2 - \frac{2\rho_k\mu^\star e^{-\rho_k\mu^\star}}{1-e^{-\rho_k\mu^\star}}\right) \left( \frac{1- e^{-\rho_k\mu^\star} }{\mu^\star}
		\right)^2 \amp^2 + 2\pi e\sigma_k^2 }{ \exp \left( 2 - \frac{2\rho_{k-1}\mu^\star e^{-\rho_{k-1}\mu^\star}}{1-e^{-\rho_{k-1}\mu^\star}}\right) \left( \frac{1- e^{-\rho_{k-1}\mu^\star} }{\mu^\star}
		\right)^2 \amp^2 + 2\pi e\sigma_k^2 } \right), \ k\in[K].	\label{eq: 190}
\end{IEEEeqnarray}}Finally, combining the above newly derived outer bound and the inner bound in Theorem~\ref{theorem: K-user inner bound both power}, we can complete the proof.
\end{proof}

The derived bounds on the capacity region are shown in Fig.~\ref{fig: K-user bpower}, where we assume $K=3$, $\sigma_3 = 2\sigma_2=4\sigma_1$, $\alpha=0.4$, and $\mathrm{ASNR}_k=\frac{\amp}{\sigma_k}$, $k\in\{1,2,3\}$. 

\begin{figure}[ht]
	\centering
	\subfigure[Bounds on capacity region.]{
		\includegraphics[width=3.3in]{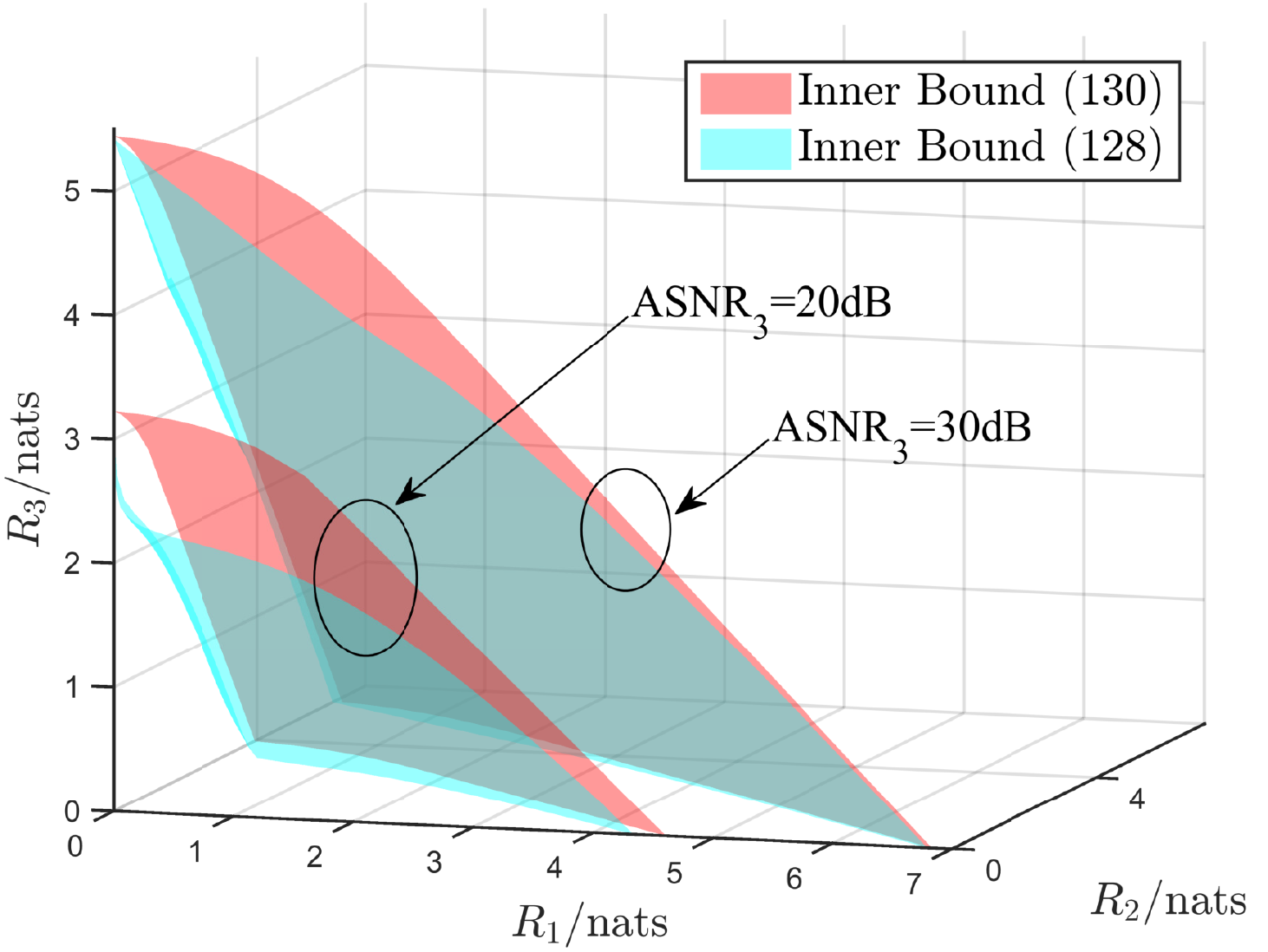}
	}\hspace{-2.5mm}		
	\subfigure[Projection on the $R_1$ and $R_2$ plane.]{
		\includegraphics[width=3.2in]{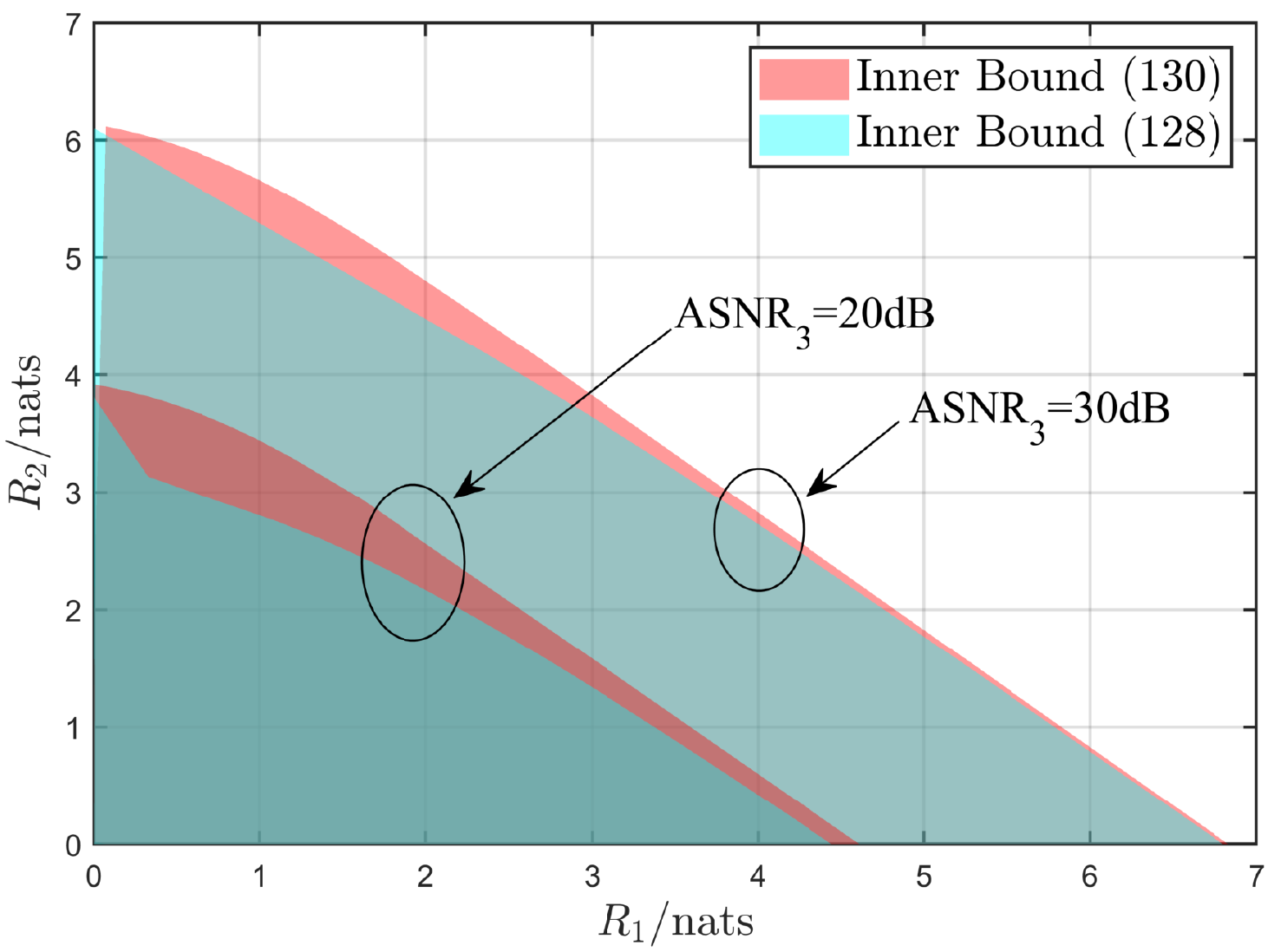}
	}
	\subfigure[Projection on the $R_2$ and $R_3$ plane.]{
		\includegraphics[width=3.2in]{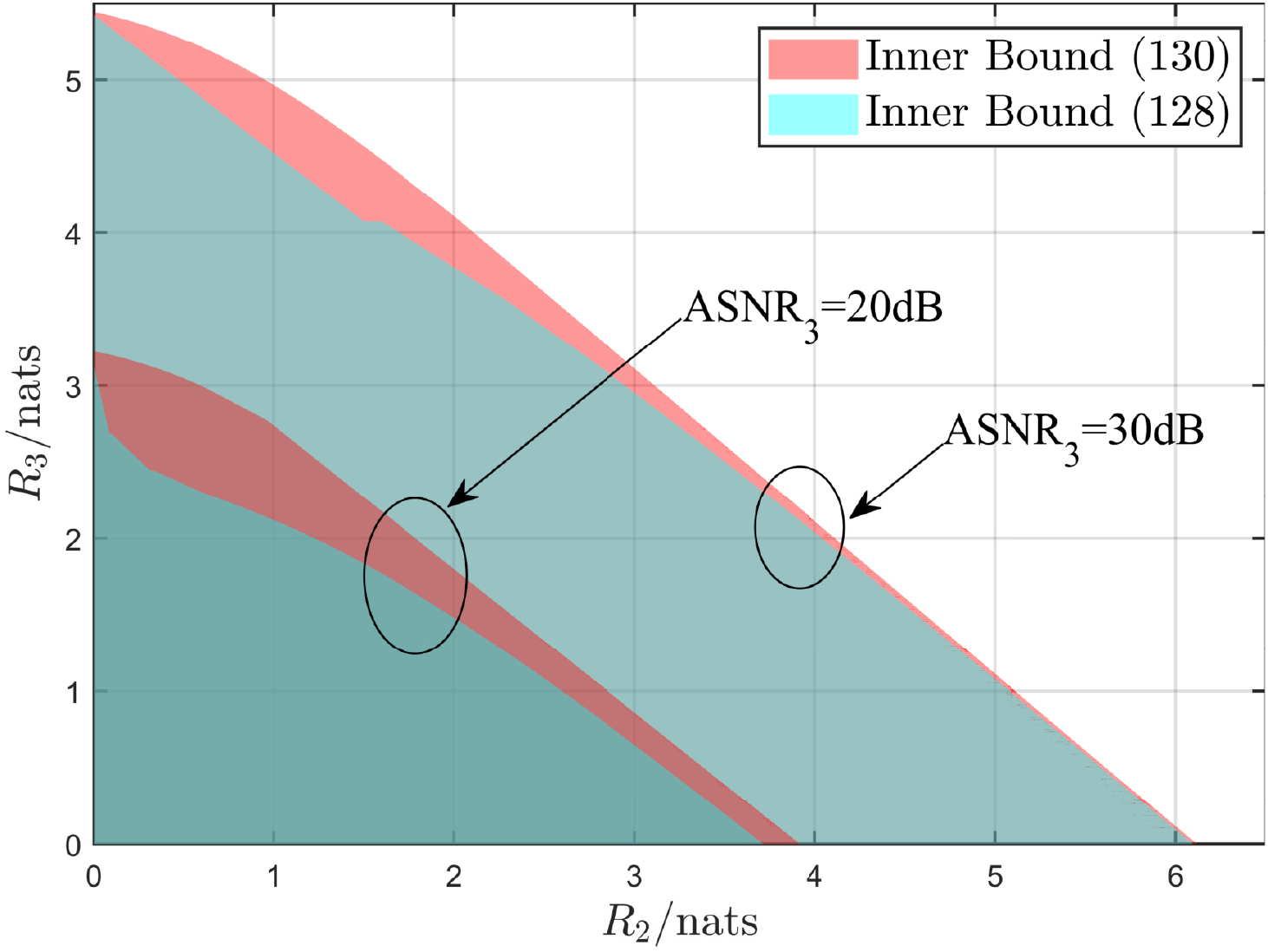}
	}\hspace{-0mm}
	\subfigure[Projection on the $R_1$ and $R_3$ plane.]{
		\includegraphics[width=3.2in]{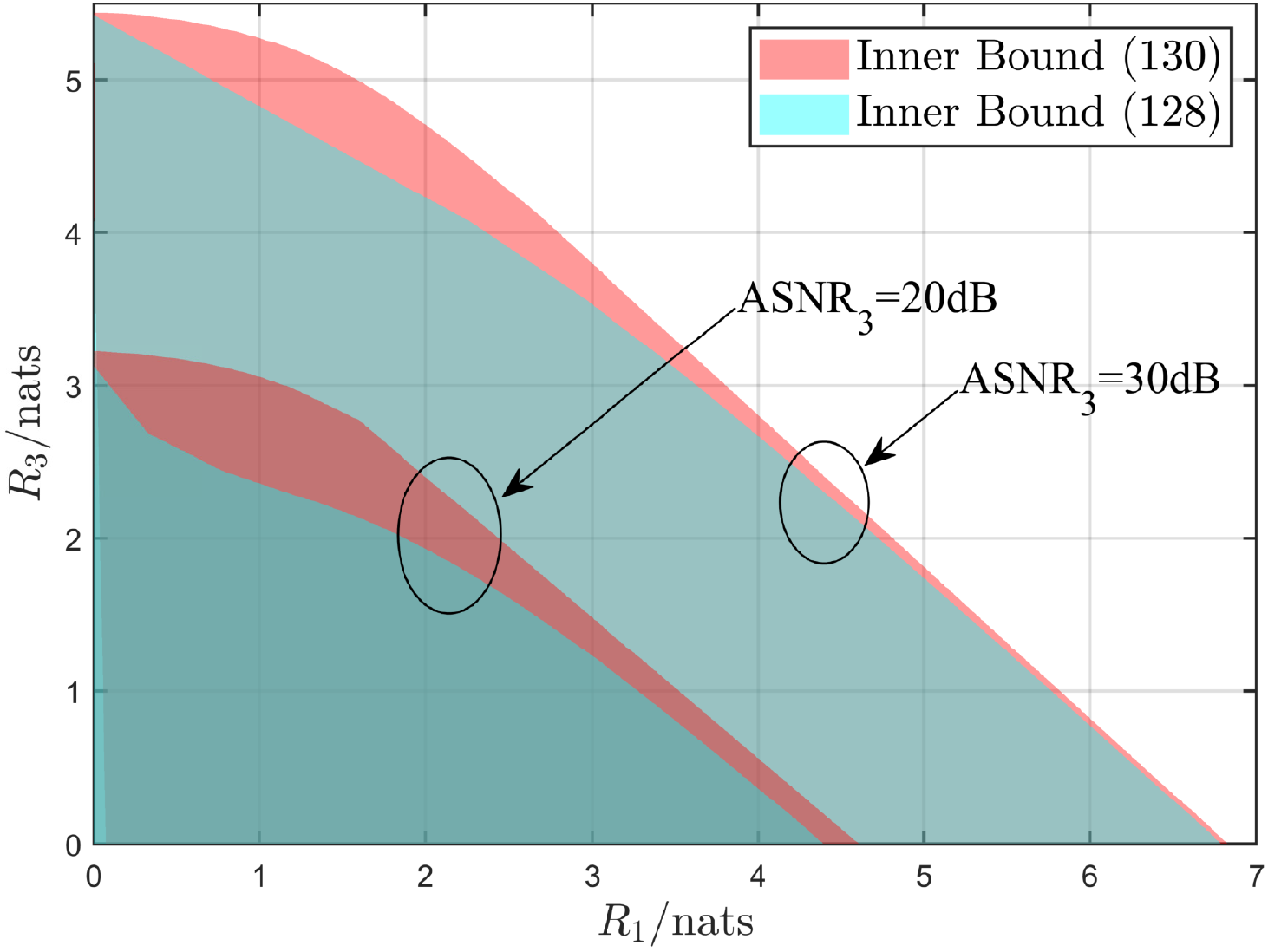}
	}	
	\caption{Bounds on capacity region of $3$-user OI-BC with peak- and average-intensity constraint when $\alpha=0.4$.}
	\vspace{-0.7cm}
	\label{fig: K-user bpower}
\end{figure}

\section{Conclusion Remarks}\label{sec: conclusion}
\label{sec:conclusion}
In this paper, we characterize the capacity region of OI-BCs. Three different input constraints are considered, i.e., (1) only peak-intensity constraint; (2) only average-intensity constraint; (3) both peak- and average-intensity constraints. We first consider two-user OI-BCs. New inner bounds are obtained by carefully designing the input for each user and adopting the SC scheme; new outer bounds on the capacity region are obtained by applying the conditional EPI. The inner and outer bounds asymptotically match at high SNR. Then we extend our results to the general $K$-user OI-BCs without loss of asymptotic optimality at high SNR. As an extension of this work, it would be interesting to study the impact of channel fading on the capacity region of OI-BCs, which could refer to \cite{Zhu2002,Bayaki2009,Yang2014}.

%

\appendices

\section{Monotonicity of $\mathsf{C}_{\textnormal{ub}}(\rho\amp,\sigma;\rho\mu)$ at High SNR}\label{app: monotonicity of C_mu(A,sigma)}
Note that at high SNR,
{\setlength\abovedisplayskip{4.85pt} 
	\setlength\belowdisplayskip{4.85pt}
\begin{IEEEeqnarray}{rCl}
	\mathsf{C}_{\textnormal{ub}}(\amp,\sigma;\mu) \doteq \frac{1}{2}\log 
	\left(1+ \exp \left(2 - \frac{2\mu e^{-\mu}}{1-e^{-\mu}}\right) \left( \frac{1- e^{-\mu} }{\mu}
	\right)^2 \frac{\amp^2 }{2\pi e\sigma^2} \right).
\end{IEEEeqnarray}}We denote a function $g_1(\rho)=\mathsf{C}_{\textnormal{ub}}(\rho\amp,\sigma;\rho\mu)$, i.e.,
{\setlength\abovedisplayskip{4.85pt} 
	\setlength\belowdisplayskip{4.85pt}
\begin{IEEEeqnarray}{rCl}
	g_1(\rho) 
	&=& \frac{1}{2}\log 
	\left(1+ \exp \left( 2 - \frac{ 2\rho\mu e^{-\rho\mu} }{1-e^{-\rho\mu}}\right) \left( \frac{1- e^{-\rho\mu} }{\mu}
	\right)^2 \frac{ \amp^2 }{2\pi e\sigma^2} \right),\nonumber\\
	&&\qquad\qquad\qquad\qquad \rho\in[0,1],\quad \mu>0,\quad \amp>0,\quad \sigma>0.
\end{IEEEeqnarray}}Fixed $\mu$, $\amp$, and $\sigma$, we find that the monotonicity of $g_1(\rho)$ is equivalent to that of the following functions:
{\setlength\abovedisplayskip{4.85pt} 
\setlength\belowdisplayskip{4.85pt}
\begin{align}
	g_2(x) &= \exp \left( 2 - \frac{2x e^{-x}}{1-e^{-x}}\right) \times (1-e^{-x})^2, \ x\geq0,\\
	g_3(x) &= \exp \left( 1 - \frac{x e^{-x}}{1-e^{-x}} \right) \times (1-e^{-x}), \ x\geq0,
\end{align}}where $g_2(x) = \bigl(g_3(x)\bigr)^2$. To prove that $g_3(x)$ is monotonically increasing with respect to $x$, we only need to analyze the following function:
{\setlength\abovedisplayskip{4.85pt} 
\setlength\belowdisplayskip{4.85pt}
\begin{align}
	g_4(x) &= 1 - \frac{x e^{-x}}{1-e^{-x}}, \quad x\geq0
\end{align}}The derivation of $g_4(x)$ is given by
{\setlength\abovedisplayskip{4.85pt} 
\setlength\belowdisplayskip{4.85pt}
\begin{align}
	g_4^\prime(x) 	
	&= \frac{ xe^{-x} + e^{-2x} - e^{-x} }{ (1-e^{-x})^2 } ,\quad x\geq0 
\end{align}}As $x$ increases form $0$ to $+\infty$, the numerator $xe^{-x} + e^{-2x} - e^{-x}$ first increases and then decreases with respect to $x$ . Since $g_4^\prime(0) =0$ and $g_4^\prime(+\infty) =0$. Thus, wen can obtain that
{\setlength\abovedisplayskip{4.85pt} 
\setlength\belowdisplayskip{4.85pt}
\begin{align}
	g_4^\prime(x) 
	\geq 0 ,\quad x\geq0 
\end{align}}Finally, we can conclude that $g_1(\rho) $ is monotonically increasing with respect to $\rho$. Equivalently, $\mathsf{C}_{\textnormal{ub}}(\rho\amp,\sigma;\rho\mu)$ is monotonically increasing with respect to $\rho$ at high SNR.

\section{Proof of Eq. (111)}
\label{app: proof of eq. 2}
Combined with Theorem~\ref{theorem: K-user outer bound only peak}, we have
{\setlength\abovedisplayskip{4.85pt} 
	\setlength\belowdisplayskip{4.85pt}
\begin{IEEEeqnarray}{rCl}
	R_{K}^{\textnormal{out}}
	&=& \frac{1}{2} \log \left( \frac{ \sigma_K^2 + \sigma_K^2( e^{ 2 \mathsf{C}_{\textnormal{ub}}({\rho}_K\amp,\sigma_K) } -1 ) }{ \sigma_K^2 + \sigma_K^2( e^{ 2 \mathsf{C}_{\textnormal{ub}}({\rho}_{K-1}\amp,\sigma_K) } -1 ) }\right)\\
	&\doteq& \frac{1}{2} \log\left( 1+ \frac{\amp^2}{2\pi e \sigma_K^2} \right) 
	- \frac{1}{2}\log \left(  1 + \left( e^{ 2 \mathsf{C}_{\textnormal{ub}}({\rho}_{K-1}\amp,\sigma_K) } -1 \right) \right), \label{eq: 134}.
\end{IEEEeqnarray}}where \eqref{eq: 134} follows from the single-user capacity result in Lemma~\ref{lemma: 1-user capacity only peak} and ${\rho}_K=1$. To characterize the relationship between $R_{K}^{\textnormal{out}}$ and $R_{k}^{\textnormal{out}}$, $k\in[K-1]$, we assume 
{\setlength\abovedisplayskip{4.85pt} 
	\setlength\belowdisplayskip{4.85pt}
\begin{IEEEeqnarray}{rCl}
	 e^{ 2 \mathsf{C}_{\textnormal{ub}}({\rho}_k\amp,\sigma_K) } -1 \doteq \sum_{m=1}^{k} \frac{\sigma_m^2}{\sigma_K^2}\left( e^{2R_{m}^{\textnormal{out}}} -1 \right) \prod_{n=m+1}^k e^{ 2R_{n}^{\textnormal{out}}}, \quad k\in[K-1].\label{eq: 135}
\end{IEEEeqnarray}}We resort to mathematical induction to prove \eqref{eq: 135}. Recall that
{\setlength\abovedisplayskip{4.85pt} 
	\setlength\belowdisplayskip{4.85pt}
\begin{IEEEeqnarray}{rCl}
	R_{1}^{\textnormal{out}} 
	&=& \frac{1}{2} \log \left( \frac{ \sigma_1^2 + \sigma_K^2( e^{ 2 \mathsf{C}_{\textnormal{ub}}({\rho}_1\amp,\sigma_K) } -1 ) }{ \sigma_1^2 + \sigma_K^2( e^{ 2 \mathsf{C}_{\textnormal{ub}}({\rho}_0\amp,\sigma_K) } -1 ) }\right)\\
	&\doteq& \frac{1}{2} \log \left( 1 + \frac{ \sigma_K^2 }{ \sigma_1^2 } \left( e^{ 2 \mathsf{C}_{\textnormal{ub}}({\rho}_1\amp,\sigma_K) } -1 \right) \right).
\end{IEEEeqnarray}}Then we have
{\setlength\abovedisplayskip{4.85pt} 
	\setlength\belowdisplayskip{4.85pt}
\begin{IEEEeqnarray}{rCl}
	e^{ 2 \mathsf{C}_{\textnormal{ub}}({\rho}_1\amp,\sigma_K) } -1  
	&\doteq& \frac{\sigma_1^2}{\sigma_K^2}\left( e^{2R_{1}^{\textnormal{out}}} -1 \right) .
\end{IEEEeqnarray}}Hence, \eqref{eq: 135} is true if $k=1$. Next, we fix a particular $i\in\{1,\cdots,K-2\}$ and assume \eqref{eq: 135} is true if $k=i$, i.e.,
{\setlength\abovedisplayskip{4.85pt} 
	\setlength\belowdisplayskip{4.85pt}
\begin{IEEEeqnarray}{rCl}
	e^{ 2 \mathsf{C}_{\textnormal{ub}}({\rho}_i\amp,\sigma_K) } -1 
	&\doteq& \sum_{m=1}^{i} \frac{\sigma_m^2}{\sigma_K^2} \left( e^{2R_{m}^{\textnormal{out}}} -1 \right) \prod_{n=m+1}^i e^{ 2R_{n}^{\textnormal{out}}}, 
\end{IEEEeqnarray}}It follows that
{\setlength\abovedisplayskip{4.85pt} 
	\setlength\belowdisplayskip{4.85pt}
\begin{IEEEeqnarray}{rCl}
	R_{i+1}^{\textnormal{out}} 
	&\doteq& \frac{1}{2} \log \left( \frac{ \sigma_{i+1}^2 + \sigma_K^2( e^{ 2 \mathsf{C}_{\textnormal{ub}}({\rho}_{i+1}\amp,\sigma_K) } -1 ) }{ \sigma_{i+1}^2 + \sigma_K^2 \sum_{m=1}^{i} \frac{\sigma_m^2}{\sigma_K^2} \left( e^{2R_{m}^{\textnormal{out}}} -1 \right) \prod_{n=m+1}^i e^{ 2R_{n}^{\textnormal{out}}} }\right),
\end{IEEEeqnarray}}and 
{\setlength\abovedisplayskip{4.85pt} 
	\setlength\belowdisplayskip{4.85pt}
\begin{IEEEeqnarray}{rCl}
	e^{ 2 \mathsf{C}_{\textnormal{ub}}({\rho}_{i+1}\amp,\sigma_K) } -1 
	&\doteq& e^{2R_{i+1}^{\textnormal{out}}} \times \left( \frac{\sigma_{i+1}^2}{\sigma_K^2} + \sum_{m=1}^{i} \frac{\sigma_m^2}{\sigma_K^2} \left( e^{2R_{m}^{\textnormal{out}}} -1 \right) \prod_{n=m+1}^i e^{ 2R_{n}^{\textnormal{out}}} \right) - \frac{\sigma_{i+1}^2}{\sigma_K^2}\\
	&=& \sum_{m=1}^{i+1} \frac{\sigma_m^2}{\sigma_K^2} \left( e^{2R_{m}^{\textnormal{out}}} -1 \right) \prod_{n=m+1}^{i+1} e^{ 2R_{n}^{\textnormal{out}}}.
\end{IEEEeqnarray}}Therefore, \eqref{eq: 135} is also true if $k=i+1$. Finally, by mathematical induction, we conclude that \eqref{eq: 135} holds for every $k\in[K-1]$.

Substituting \eqref{eq: 135} into \eqref{eq: 134}, the relationship between $R_{K}^{\textnormal{out}}$ and $R_{k}^{\textnormal{out}}$, $k\in[K-1]$, can be characterized by 
{\setlength\abovedisplayskip{4.85pt} 
	\setlength\belowdisplayskip{4.85pt}
\begin{IEEEeqnarray}{rCl}
	R_{K}^{\textnormal{out}} 
	&\doteq& \frac{1}{2} \log\left( 1+ \frac{\amp^2}{2\pi e \sigma_K^2} \right) 
	- \frac{1}{2}\log \left(  1 + \sum_{m=1}^{K-1} \frac{\sigma_m^2}{\sigma_K^2}\left( e^{2R_{m}^{\textnormal{out}}} -1 \right) \prod_{n=m+1}^{K-1} e^{ 2R_{n}^{\textnormal{out}}} \right),\quad
\end{IEEEeqnarray}}which concludes the proof.

\footnotesize
\normalsize

\end{document}